\documentclass[11pt]{article}
\usepackage{extras}
\usepackage{amsmath,amssymb} 
\usepackage{nicefrac}
\usepackage{tabularx} 

\newcommand{\nfrac}{\nicefrac}
\newcommand{\defeq}{\stackrel{\textup{def}}{=}}
\long\def\symbolfootnote[#1]#2{\begingroup%
\def\thefootnote{\fnsymbol{footnote}}\footnote[#1]{#2}\endgroup}

\newcommand{\CL}{\mbox{\footnotesize ${\mathcal L}$}}

\newcommand{\CP}{\mbox{${\mathcal P}$}}

\newcommand{\R}{\mathbb{R}}

\newcommand{\vv}{\mbox{\boldmath $v$}}
\newcommand{\yy}{\mbox{\boldmath $y$}}
\newcommand{\zz}{\mbox{\boldmath $z$}}

\newcommand{\pp}{{\mbox{\boldmath $p$}}}

\newcommand{\qq}{\mbox{\boldmath $q$}}
\newcommand{\rr}{\mbox{\boldmath $r$}}
\newcommand{\xx}{\mbox{\boldmath $x$}}

\DeclareMathOperator*{\argmax}{arg\,max}

\newcommand{\ZZ}{{\mbox{$\mathbb Z$}}}
\newcommand{\QQ}{{\mbox{$\mathbb Q$}}}
\renewcommand{\R}{{\mbox{$\mathbb R$}}}
\newcommand{\Real}{{\mbox{$\mathbb R$}}}
\newcommand{\bb}{\mbox{\boldmath $b$}}

\newcommand{\cc}{\mbox{\boldmath $c$}}
\newcommand{\uu}{\mbox{\boldmath $u$}}
\renewcommand{\ll}{\mbox{\boldmath $\lambda$}}
\renewcommand{\l}{\mbox{$\lambda$}}
\renewcommand{\b}{\mbox{$\beta$}}
\newcommand{\bbeta}{\mbox{\boldmath $\beta$}}
\newcommand{\lcp}{\mbox{LCP}}
\newcommand{\tx}{\mbox{$\tilde{x}$}}
\newcommand{\ty}{\mbox{$\tilde{y}$}}
\newcommand{\tA}{\mbox{$\tilde{A}$}}
\newcommand{\tB}{\mbox{$\tilde{B}$}}
\renewcommand{\qq}{{\mbox{\boldmath $q$}}}
\newcommand{\ee}{{\mbox{\boldmath $e$}}}
\newcommand{\txx}{\mbox{\boldmath $\tilde{x}$}}
\newcommand{\tzz}{\mbox{\boldmath $\tilde{z}$}}
\newcommand{\tyy}{\mbox{\boldmath $\tilde{y}$}}
\newcommand{\ones}{\mbox{\boldmath $1$}}
\newcommand{\zeros}{\mbox{\boldmath $0$}}
\newcommand{\D}{{\mbox{$\Delta$}}}
\newcommand{\dpo}{\mbox{$\Delta^+_1$}}
\newcommand{\dmo}{\mbox{$\Delta^-_1$}}
\newcommand{\dpt}{\mbox{$\Delta^+_2$}}
\newcommand{\dmt}{\mbox{$\Delta^-_2$}}

\newcommand{\Djpo}{\mbox{$\Delta^{j+}_1$}}
\newcommand{\Djmo}{\mbox{$\Delta^{j-}_1$}}
\newcommand{\Djpt}{\mbox{$\Delta^{j+}_2$}}
\newcommand{\Djmt}{\mbox{$\Delta^{j-}_2$}}

\def\CP{{\mathcal P}}

\def\CP{{\mathcal P}}

\date{}
\title{Constant Rank Bimatrix Games are PPAD-hard}
\author{\large Ruta Mehta\thanks{Supported by NSF Grant CCF-1216019. 
}\\
\footnotesize{College of Computing, Georgia Institute of Technology.}\\
\footnotesize{Email: rmehta@cc.gatech.edu}}

\begin{document}
\maketitle

\thispagestyle{empty}

\begin{abstract}
The rank of a bimatrix game $(A,B)$ is defined as $rank(A+B)$.  Computing a Nash equilibrium (NE) of a rank-$0$, i.e.,
zero-sum game is equivalent to linear programming (von Neumann'28, Dantzig'51).  In 2005, Kannan and Theobald gave an FPTAS
for constant rank games, and asked if there exists a polynomial time algorithm to compute an exact NE. Adsul et al. (2011)
answered this question affirmatively for rank-$1$ games, leaving rank-$2$ and beyond unresolved.

In this paper we show that NE computation in games with rank $\ge3$, is PPAD-hard, settling a decade long open problem.
Interestingly, this is the first instance that a problem with an FPTAS turns out to be PPAD-hard.  Our reduction bypasses
graphical games and game gadgets, and provides a simpler proof of PPAD-hardness for NE computation in bimatrix games. In
addition, we get: 

\begin{itemize}
\item An equivalence between $2D$-Linear-FIXP and PPAD, improving a result by
Etessami and Yannakakis (2007) on equivalence between Linear-FIXP and PPAD.
\item NE computation in a bimatrix game with convex set of Nash equilibria is as hard as solving a simple stochastic game
\cite{condon}.
\item Computing a symmetric NE of a symmetric bimatrix game with rank $\ge 6$ is PPAD-hard.
\item Computing a $\frac{1}{poly(n)}$-approximate fixed-point of a (Linear-FIXP) piecewise-linear function is PPAD-hard.
\end{itemize}

The status of rank-$2$ games remains unresolved.

\end{abstract}
\newpage
\setcounter{page}{1}

\section{Introduction}
Two player, finite, non-cooperative games constitute the most simple and fundamental model within game theory
\cite{myerson}, and have been studied extensively for their computational and structural properties.  Such a game can be
represented by two payoff matrices $(A,B)$, one for each player, and therefore are also known as bimatrix games. Von Neumann
(1928) showed that in games where one player's loss is the other player's gain ($B=-A$, zero-sum), the min-max strategies
are stable \cite{neumann}. This turned out to be equivalent to linear programming (LP) \cite{dantzig,adler_zerosum} and
therefore polynomial-time computable. In 1950, John Nash \cite{nash} extended this notion to formulate an equilibrium
concept, and proved its existence for finite multi-player games.  It has since been named Nash equilibrium (NE) and is
perhaps the most important and well-studied solution concept in game theory. 

The classical Lemke-Howson algorithm (1964) \cite{LH}, to compute Nash equilibrium in general bimatrix games, performs very
well in practice. However it may take exponential time in the worst case \cite{sav}. 
Other methods that followed \cite{lemke,ET1} are also similar in nature \cite{ann,GPS}, and a complexity theoretic study of
the problem was called for. Henceforth, by $2$-Nash we mean computing a Nash equilibrium of a bimatrix game.

The complexity class NP is not applicable for $2$-Nash, because an equilibrium is guaranteed to exist \cite{nash}.  However,
computing a special kind of NE, for numerous special properties, has been shown to be NP-complete \cite{GZ,CS}.  In 1994
Papadimitriou introduced complexity class PPAD \cite{papa}, {\em Polynomial Parity Argument for Directed} graph, for
problems with path following argument for existence, like Sperner's lemma \cite{sperner}. He showed that $2$-Nash, among
many other problems, is in PPAD. After more than a decade, the problem was shown to be PPAD-hard in a remarkable series of
works \cite{DGP,CDT}.  Chen et. al. \cite{CDT} showed that even $\frac{1}{poly(n)}$-approximation of $2$-Nash is PPAD-hard,
{\em i.e.}, if there is a fully polynomial-time approximation scheme (FPTAS) for $2$-Nash then PPAD=P.  This was followed by
PPAD-hardness results for special classes of bimatrix games, like sparse games \cite{CDT_sparse} and win-lose games
\cite{winlose}, and their approximation were also shown to be PPAD-hard.

On the positive side, polynomial-time algorithms were developed for many special classes of games; 
see Section \ref{sec.rw} for an overview of previous results. 
Among these, one of the most significant is the class of constant rank games defined by
Kannan and Theobald (2005) - rank of game $(A,B)$ is defined as $rank(A+B)$.  They gave an FPTAS for constant rank
games,\footnote{$O(\nfrac{L}{\epsilon})^k poly(n)$ time algorithm to compute an $\epsilon$-approximate Nash equilibrium in a
rank-$k$ $n \times n$ game of bit size $L$.} and asked if there is an efficient algorithm to compute an exact NE in these
games. Note that, rank-$0$ are zero-sum games, and therefore are polynomial-time solvable.  For rank-$1$ games, Adsul et.
al. \cite{AGMS} gave a polynomial time algorithm, by reducing the problem to $1$-dimensional fixed-point, however rank-$2$
and beyond remained unresolved. 

In this paper we show that NE computation in games with rank $\ge3$ is PPAD-hard, settling a decade long open problem. Since
there is an FPTAS for constant rank games, this result comes as a surprise, because until now whenever a problem, in games
or markets, was shown to be PPAD-hard, so was its approximation (i.e., no FPTAS unless PPAD=P)
\cite{CDT,CDT_sparse,winlose,Chen.plc,kintali}.

To obtain the result, we reduce $2D$-Brouwer, a two dimensional discrete fixed point problem which is known to be PPAD-hard
\cite{CD}, to a rank-$3$ game. The reduction is done in two steps. First we reduce $2D$-Brouwer to $2D$-Linear-FIXP;
Linear-FIXP \cite{EY07} is a class of fixed-point problems with polynomial piecewise-linear functions, and $kD$-Linear-FIXP
is its subclass consisting of $k$-dimensional fixed-point problems. 
In the second step, we reduce an instance of $2D$-Linear-FIXP to a rank-$3$ bimatrix game, such that a
linear function of Nash equilibrium strategies of the resulting game gives fixed-points of the $2D$-Linear-FIXP instance.

Our reduction completely bypasses the machinery of graphical games and game gadgets, central to the previous approaches, and
instead exploits relations between LPs, linear complementarity problems (LCPs) and bimatrix games.  Such a conceptual leap
seems to be necessary to show hardness of constant rank games, since the game gadgets used previously inherently give rise
to higher rank games. Our approach also provides a simpler proof for PPAD-hardness of $2$-Nash, and may be of
independent interest 
to show hardness for other problems, and to understand connections between
parameterized LPs and bimatrix games. We can achieve further simplification by avoiding even the parameterized LP, but the
resulting game turns out to be of high rank.

Apart from the hardness of constant rank games, a number of results follow as corollaries from our reduction.  The first
step shows PPAD-hardness of $2D$-Linear-FIXP and thereby improves the equivalence result Linear-FIXP = PPAD of Etessami and
Yannakakis to $2D$-Linear-FIXP = PPAD. This also implies $2D$-Linear-FIXP = Linear-FIXP; in other words the class of
Linear-FIXP remains unchanged even when functions are restricted to two dimensions. Since, an instance of 
$1D$-Linear-FIXP can be solved in polynomial time using binary search, 
our result establishes a {\em dichotomy} between $1D$ and $kD$, $k\ge2$ Linear-FIXP
problems; the former are in P and the latter are PPAD-complete.

Our approach can be extended to reduce $kD$-Brouwer to $kD$-Linear-FIXP to rank-$(k+1)$ games, where the reduction from
$kD$-Linear-FIXP to rank-$(k+1)$ games (almost) preserves the number of solutions. Using this, together
with a result from \cite{EY07}, we show that bimatrix games with convex set of NE are no easier. In fact they are as hard as solving
simple stochastic games, which are known to be in NP $\cap$ coNP \cite{condon}, however despite significant efforts its exact
complexity remains open \cite{condon_alg,AM}. Further, we can show that computing weak\footnote{Vector $\xx$ is a {\em weak}
$\epsilon$-approximate fixed-point of function $f$ if $\|\xx-f(\xx)\|_\infty \le \epsilon$} $\frac{1}{poly(n)}$-approximate
fixed-point of a function in Linear-FIXP is also PPAD-hard. It will be interesting to see if this can be
extended to show hardness of approximation in $2$-Nash.

Since NE computation in a rank-$k$ game can be reduced to computing symmetric NE of a symmetric game with rank-$2k$ 
\cite{agt.ch2}, we get that computing symmetric Nash equilibria in symmetric games with rank $\ge 6$ is PPAD-hard.  Again
computing symmetric NE in symmetric rank-$0$ games can be solved using LP, and for rank-$1$ games recently Mehta et.  al.
\cite{MVY} gave a polynomial-time algorithm. This leaves the status of symmetric games with rank between $2$ and $5$
unresolved. Also the status of rank-$2$ bimatrix games remains unresolved.

\subsection{Overview of the Reduction}
In this section we explain the main ideas behind the reductions: from $2D$-Brouwer to $2D$-Linear-FIXP, and then to 
rank-$3$ game. We start with a brief description of $2D$-Brouwer and Linear-FIXP problems. 

$2D$-Brouwer is a class of $2$-dimensional discrete fixed-point problems, known to be PPAD-hard \cite{CD2D}.  An instance of
$2D$-Brouwer consists of a grid $G_n=\{0,\dots,2^n-1\}\times\{0,\dots,2^n-1\}$ and a valid coloring function
$g:G_n\rightarrow\{0,1,2\}$ which satisfies some boundary conditions, and thereby ensures existence of a trichromatic unit
square in the grid.\footnote{This is similar to the Sperner's lemma} The problem is to find one such trichromatic square
(see Section \ref{sec.2db} for details).  Function $g$ is specified by a Boolean circuit $C^b$ with $2n$ input bits; $n$
bits to represent each of the two co-ordinate of a grid point.\footnote{We use super-script $b$ to differentiate Boolean
circuits from Linear-FIXP circuits that will follow.}

Linear-FIXP \cite{EY07} is a class of fixed-point problems with polynomial piecewise-linear functions.  A function $F:[0,\
1]^n$ $\rightarrow[0,\ 1]^n$ in Linear-FIXP is defined by a circuit, say $C$, with $n$ real inputs and outputs, and
$\{\max,+,*\zeta\}$ operations, where $*\zeta$ is multiplication by a rational constant (see Section \ref{sec.lf} for
details).  Such a function has rational fixed-points of size polynomial in the input size \cite{EY07}.  We denote the class
of $k$-dimensional fixed-point problems in Linear-FIXP by $kD$-Linear-FIXP. 

Given circuit $C^b$ of a $2D$-Brouwer instance, in Section \ref{sec.2dlf} we construct a $2D$-Linear-FIXP circuit $C$ such
that all the fixed-points of the function $F$ defined by $C$ are in trichromatic unit squares of the grid $G_n$.  It is easy
to simulate $C^b$ in $C$ by replacing $\land$, $\lor$ and $\lnot$ with $\min, \max$ and $(1-x)$ respectively, if input to
this simulation is guaranteed to be Boolean. To guarantee this, we need to extract
bit representation of $\lfloor \pp \rfloor$, for a
$\pp \in [0,\ 2^n-1]^2$.  Since, {\em floor} is a discontinuous function it can not be simulated using Linear-FIXP circuit, whose
operations can generate only continuous functions. 
However, we design a bit extraction gadget which does the job for almost all the points, except those that are close to the
boundary of unit squares of $G_n$. Finally, using a sampling lemma similar to that of \cite{CDT} we ensure that the
fixed-points of the function defined by the resulting circuit $C$ are always in trichromatic unit squares of the grid $G_n$,
and we get,

\begin{theorem}[Informal]\label{thm1}
Computing a fixed-point of a Linear-FIXP instance with $k$ inputs and $k$ outputs, with $k>1$, is PPAD-hard.
In other words, $2D$-Linear-FIXP=PPAD.
\end{theorem}

Etessami and Yannakakis \cite{EY07} showed that Linear-FIXP $=$ PPAD. Theorem \ref{thm1} improves this to $2D$-Linear-FIXP =
PPAD, and in turn we get Linear-FIXP = $2D$-Linear-FIXP, {\em i.e.,} 
fixed-point problems with polynomial piecewise-linear functions in constant (two) dimension are as hard as those in
$n$-dimension.

Next, we reduce the fixed-point computation of a $kD$-Linear-FIXP instance to Nash equilibrium computation in a rank-$(k+1)$
game (see Section \ref{sec.r3}). Let $\ll=(\l_1,\dots,\l_k)$ denote the $k$ inputs of circuit $C$ of the given
$kD$-Linear-FIXP instance.  First we replace circuit $C$ by a parameterized linear program $LP(\ll)$, so that circuit
evaluation for a given input is same as solving the LP. 

This is done as follows: There is an ordering among $\max$ gates since $C$ forms a DAG. Suppose $x_i$ captures the output of
the $i^{th}$ $\max$ gate.  Since the $+$ and $*\zeta$ operations of circuit $C$ generates only linear expressions, for
$x_j=\max\{L,R\}$, $L$ and $R$ both are linear expression in $x_1,\dots,x_{j-1}$ and $\ll.$ Further, this $\max$ operation
is equivalent to $x_j\ge L, x_j\ge R, (x_j-L)(x_j-R)=0$. The first two linear conditions define the feasible region of
LP, where $x_j$s are variables and $\l_i$s are parameters. Note that, r.h.s. of the constraints of LP is parameterized by
$(\l_1,\dots,\l_k)$, and the constraint matrix is lower triangular. 
Using this property we show that $\exists \cc$ such that for all $\ll$, $\min: \cc^T\xx$ over this
feasible region will ensure the quadratic constraints as well for each $\max$ gate. This gives the $LP(\ll)$ which can
replace circuit $C$.

Since, primal-dual feasibility, and complementary slackness characterizes solutions of an LP, LP is a special case of linear
complementarity problem (LCP). Using this connection for $LP(\ll)$, we construct an $\lcp_C$ whose solutions exactly capture
the fixed point of the given $kD$-Linear-FIXP instance (Section \ref{sec.lcp}). Further, the matrix of the LCP turns out to
be off-block-diagonal, with the two blocks in off-diagonal adding up to a rank-$k$ matrix. Finally, using the fact that the
LCP capturing Nash equilibria of a bimatrix game also has a off-block-diagonal matrix, we construct a bimatrix game, whose Nash
equilibria are in one-to-one correspondence with the solutions of $\lcp_C$. The rank of the resulting game turns out to be
$(k+1)$, and one of its payoff matrix is upper-triangular. 

\begin{theorem}[Informal]\label{thm1}
Nash equilibrium computation in bimatrix games with rank $\ge 3$ is PPAD-hard, even when one of the payoff matrix is lower/upper triangular.
\end{theorem}

Theorem \ref{thm1}, together with the reduction from $2$-Nash to symmetric
 $2$-Nash \cite{agt.ch2}, implies that computing symmetric NE of a symmetric game with rank $\ge 6$ is PPAD-hard.
Further, this gives a simpler proof of PPAD-hardness of $2$-Nash, {\em i.e.,} without using the graphical games and game
gadgets.  

In Section \ref{sec.simple} we further simplify this proof by avoiding the parameterized LPs as well, where we
first construct a symmetric game whose symmetric NE are in one-to-one correspondence with the fixed-points.  As consequences
we get that Nash equilibrium computation in bimatrix games with convex set of Nash equilibria (Corollary \ref{cor.ssg}), and
computing a unique symmetric NE of a symmetric game (Corollary \ref{cor.ssg_sym}), both are as hard as solving a simple
stochastic game, since the latter reduces to finding a unique fixed-point of a Linear-FIXP problem \cite{EY07}.

In Section \ref{sec.approx} we extend the first step of the reduction, to reduce $kD$-Brouwer to $kD$-Linear-FIXP.  We show
that when an instance of $kD$-Brouwer with a $k$-dimensional grid $\{0,\dots,2^n-1\}^k$ is reduced to an instance of
$kD$-Linear-FIXP, not only exact fixed-points but also all the $\frac{1}{2^npoly(k)}$-approximate fixed-points are in
panchromatic unit cube of the grid. Chen et al. \cite{CDT} showed that a class of $kD$-Brouwer with $n=3$ is PPAD-hard,
where $k$ is an input parameter and not a constant. Therefore, we get that $\frac{1}{poly(\CL)}$-approximation of
Linear-FIXP is PPAD-hard (Theorem \ref{thm.approx}), where $\CL$ is the size of the input instance.

It will be interesting to extend this result to bimatrix games using the reduction of Section \ref{sec.simple}, and thereby
getting a simpler proof of inapproximability in 2-Nash as well. 
Importantly, our work leaves the status of rank-$2$ games, and symmetric games with rank between $2$ and $5$, unresolved. 

\subsection{Related Work}\label{sec.rw}
Efficient algorithms have been designed for many special classes of bimatrix games.  Lipton et. al. \cite{lip} gave a
pseudo-polynomial time algorithm, which remains the best known bound till now.  In addition, they gave a polynomial time
algorithm for games where $\max\{rank(A), rank(B)\}$ is a constant. Later Garg et. al. \cite{GJM} improved it to
$\min\{rank(A),rank(B)\}$ being constant. Note that, these classes are restrictive and do not capture even all of zero-sum
games.  For random games, B\'{a}r\'{a}ny et. al. \cite{BVV} showed that there exists a NE with support size $2$ with
$O(1-\nfrac{1}{\log n})$ probability, and using this gave an algorithm which is efficient with high probability.  A game is
called win-lose game, if all the entries of $A$ and $B$ are either zero or one. Chen et. al. \cite{CDT_sparse} gave a
polynomial-time algorithm for win-lose sparse games, and Addario-Berry et. al. \cite{AOV} gave one for win-lose planar
games.

Many algorithms are designed to achieve constant factor approximation for $2$-Nash \cite{DMP,BBM,TS}; the best known factor
till now is $0.3393$ due to Tsaknakis and Spirakis \cite{TS}.  Although designing a polynomial time approximation scheme
(PTAS) remains open, PTASs were designed for special classes, like Daskalakis and Papadimitriou \cite{DP_ptas} gave one for
sparse games and games whose equilibria are guaranteed to have small-$O(1/n)$-values, and Alon et. al. \cite{ALSV} gave a
PTAS for games with rank-$(\log{n})$.

\section{Preliminaries}\label{sec.prel}
To show the hardness of rank-$3$ games, we start with $2D$-Brouwer, reduce it to Linear-FIXP and then to a bimatrix game.
In this section we discuss each of these problems separately.  First we describe a characterization of Nash equilibria in
bimatrix games, and the class of $2D$-Brouwer problems.  Both the problems are known to be PPAD-complete \cite{CDT,CD2D}.
Next, we describe Linear-FIXP, 
\cite{EY07}, and define a subclass called $kD$-Linear-FIXP.  
\medskip

\noindent{\bf Notations:} All the vectors are in bold-face letters, and are considered as column vectors. To denote a row
vector we use $\xx^T$. The $i^{th}$ coordinate of the vector $\xx$ is denoted by $x_i$. $\ones$ and $\zeros$ represent all
ones and all zeros vector respectively of appropriate dimension. We use $[n]$ to denote the set $\{1,\dots,n\}$.

\subsection{Bimatrix games and Nash equilibrium}\label{sec.bimatrix}
A bimatrix game is a two player game, each player having finitely many pure strategies (moves).  Let $S_i,\ i=1,2$ be the
set of strategies of player $i$, and let $m \defeq |S_1|$ and $n\defeq|S_2|$.  Then such a game can be represented by two
payoff matrices $A$ and $B$, each of $m\times n$ dimension. If the first player plays strategy $i$ and the second plays $j$,
then the payoff of the first player is $A_{ij}$ and that of the second player is $B_{ij}$.  Note that the rows of these
matrices correspond to the strategies of the first player and the columns to the strategies of second player. 

Players may randomize among their strategies; a randomized play is called a {\em mixed strategy}.
The set of mixed strategies for the first player is $X=\{\xx=(x_1,\dots,x_m)\ |\ \xx\ge
0, \sum_{i=1}^m x_i=1\}$, and for the second player is $Y=\{\yy=(y_1,\dots, y_n)\ |\ \yy\ge 0, \sum_{j=1}^n y_j=1\}$. 
By playing $(\xx,\yy) \in X \times Y$ we mean strategies are picked independently at random as per $\xx$ by the
first-player and as per $\yy$ by the second-player. 
Therefore the expected payoffs of the first-player and second-player are, respectively\[ \sum_{i,j} A_{ij} x_i y_j=\xx^TA\yy\ \ \ \
\mbox{ and }\ \ \ \ \sum_{i,j}B_{ij}x_iy_j = \xx^TB\yy\]

\begin{definition}{(Nash Equilibrium \cite{agt.bimatrix})}
A strategy profile is said to be a Nash equilibrium strategy profile (NESP) if no player achieves a better payoff by a
unilateral deviation \cite{nash}. Formally, $(\xx,\yy) \in X\times Y$ is a NESP iff $\forall \xx' \in X,\
\xx^TA\yy \geq \xx'^TA\yy$ and $\forall \yy' \in Y,\  \xx^TB\yy \geq \xx^TB\yy'$. 
\end{definition}

Given strategy $\yy$ for the second-player, the first-player gets $(A\yy)_k$ from her $k^{th}$ strategy. Clearly, her best
strategies are $\argmax_k (A\yy)_k$, and a mixed strategy fetches the maximum payoff only if she randomizes among her best
strategies. Similarly, given $\xx$ for the first-player, the second-player gets $(\xx^TB)_k$ from $k^{th}$ strategy, and same
conclusion applies. These can be equivalently stated as the following complementarity type conditions,

\[
\begin{array}{ll}
\forall i \in S_1,\hspace{.06in} x_i>0 \ \ \Rightarrow\ \ & (A\yy)_i =\max_{k \in S_1} (A\yy)_k\\
\forall j \in S_2,\hspace{.06in} y_j>0 \ \ \Rightarrow & (\xx^TB)_j = \max_{k \in S_2} (\xx^TB)_k
\end{array}
\]

The next lemma follows from the above discussion. 

\begin{lemma}\label{lem.nash}
Strategy profile $(\xx,\yy) \in X\times Y$ is a NE of game
$(A,B)$ if and only if the following holds, where $\pi_1$ and $\pi_2$ are scalars capturing respective payoffs at $(\xx,\yy)$. 
\[
\begin{array}{c}
\forall i \in S_1, (A\yy)_i \le \pi_1;\ \ \ \ x_i((A\yy)_i-\pi_1)=0\\
\forall j \in S_2, (\xx^TB)_j \le \pi_2;\ \ \ \ y_j((\xx^TB)_j-\pi_2)=0\\
\end{array}
\]
\end{lemma}

Game $(A,B)$ is said to be symmetric if $B=A^T$. In a symmetric game the strategy sets of both the players are 
identical, i.e., $m=n$, $S_1=S_2$ and $X=Y$. We will use $n$, $S$ and $X$ to denote number of strategies, the strategy set
and the mixed strategy set respectively of the players in such a game. A Nash equilibrium profile $(\xx,\yy)\in X\times X$
is called {\em symmetric} if $\xx=\yy$. Note that at a symmetric strategy profile $(\xx,\xx)$ both the players get payoff
$\xx^TA\xx$. Using Lemma \ref{lem.nash} we get the following.

\begin{lemma}\label{lem.symnash}
Strategy profile $\xx \in X$ is a symmetric NE of game $(A,A^T)$, with payoff $\pi$ to both
players, if and only if, 
\[
\forall i \in S, (A\xx)_i \le \pi;\ \ \ \ x_i((A\xx)_i-\pi)=0
\]
\end{lemma}

The problem of computing such a Nash equilibrium strategy in bimatrix games is PPAD-complete \cite{CDT,DGP}. This also implies that
computing symmetric NE of a symmetric bimatrix game is PPAD-hard, because NE of game $(A,B)$ are in one-to-one correspondence with 
the symmetric NE of game
$(S,S^T)$ with $S=\left[\begin{array}{cc} 0 & A \\ B^T & 0\end{array}\right]$ \cite{agt.ch2}.

\subsection{2D-Brouwer}\label{sec.2db}
Let $G_{n}$ denote the two dimensional grid $\{0, \dots, 2^{n} -1\}\times \{0, \dots, 2^n-1\}$. A 3-coloring of $G_{n}$ is a
function $g$ from the vertices of $G_{n}$ to $\{0,1,2\}$. Function $g$ is said to be valid if for every vertex $(p_1, p_2)$ on the
boundary of $G_{n}$, we have
\[
\mbox{If }\ p_2=0\ \mbox{ then }\ g(\pp)=2,\ \ \mbox{ else if }\ p_2>0\ \&\ p_1=0\ \mbox{ then }\ g(\pp)=1,\ \ \mbox{ else } g(\pp)=0 
\]

Let $K_{\pp}$ denote the unit square with $\pp$ at the bottom left corner. Due to Sperner's Lemma it is known that for any
valid coloring of $g$ of $G_n$ there exists a vertex $\pp \in G_{n}$ such that vertices of $K_\pp$ have all the three colors -
trichromatic.
\medskip

\noindent{\bf 2D-Brouwer Mapping Circuit:}
Consider a Boolean circuit $C^b$ generating valid coloring on grid $G_{n}$.\footnote{We use the
definitions and terminology of \cite{CDT} to remain consistent.}
The circuit has $2n$ input bits, $n$ bits for each of the two integers representing a grid point, and $4$ output bits
$\D_1^+,\D_1^-,\D^+_2,\D^-_2$. It is a {\em valid Brouwer-mapping circuit} if the following is true:

\begin{itemize}
\item For every $\pp \in G_{n}$, the $4$ output bits of $C^b$ satisfies one of the following $3$ cases: 
\begin{itemize}
\item Case $0$: $i=1,2$, $\D^-_i=1$ and $\D^+_i=0$.
\item Case $i$, $i=1,2$: $\D^+_i=1$ and all the other $3$ bits are zero.
\end{itemize}
\item For every $\pp$ on the boundary of $G_{n}$, if $p_2=0$ then Case $2$ is satisfied, if $p_1=0$ and $p_2\neq 0$ then Case $1$ is
satisfied, and for the rest Case $0$ is satisfied. 
\end{itemize}

Such a circuit $C^b$ defines a valid color assignment $g_{C^b}: G_{n} \rightarrow \{0,1,2\}$ by setting $g_{C^b}(\pp)=i$, if
the output bits of $C^b$ evaluated at $\pp$ satisfy Case $i$.

\begin{definition}
[2D-Brouwer \cite{CD2D}] The input to the 2D-Brouwer consists of a valid Brouwer-mapping circuit $C^b$ that produces a valid
coloring on $G_{n}$. The problem is to find a point $\pp \in G_{k}$ such that $K_\pp$ is trichromatic. 
\end{definition}

The size of the given 2D-Brouwer problem is $size[C^b]$, which is \#input nodes + \#output nodes + \# gates.

The outputs of the circuit (defining a color), can also be mapped to incremental vector $(\D^+_1-\D^-_1, \D^+_2-\D^-_2)$. 
Let $\ee^i$ be the incremental vector corresponding to Case (color) $i$, 
then clearly, $\ee^0=(-1,-1), \ee^1=(1,0)$ and $\ee^2=(0,1)$. Define a discrete function $H$, such that
$H(\pp)=\pp+\ee^{g_{C^b}(\pp)}$. 
It is easy to see that if $C^b$ is a valid Brouwer-mapping circuit, then $H$ is $G_{n}\rightarrow G_{n}$, and vertices of a trichromatic
square $K_\pp$ goes in each of the $\ee^i$ direction under $H$.
Chen and Deng showed finding such a square is PPAD-hard \cite{CD2D}.

\subsection{Linear-FIXP}\label{sec.lf}
Etessami and Yannakakis \cite{EY07} defined the class FIXP to capture complexity of the exact fixed point problems with
algebraic solutions. An instance $I$ of FIXP consists of an algebraic circuit $C_I$ defining a function $F_I:[0,\ 1]^d
\rightarrow [0,\ 1]^d$, and the problem is to compute a fixed-point of $F_I$. The circuit is a finite representation of
function $F_I$ (like a formula), consisting of $\{\max, + , *\}$ operations, rational constants, and $d$ inputs and outputs.

The circuit $C_I$ is a sequence of gates $g_1,\dots,g_m$, where for $ i\in [d]$, $g_i:= \l_i$ is an input variable. For
$d< i\le d+r$, $g_i:=c_i \in \QQ$ is a rational constant, with numerator and denominator encoded in binary. For $i> d+r$ we
have $g_i=g_j \circ g_k$, where $j,k < i$ and the binary operator $\circ \in \{ \max, +, *\}$. The last $d$ gates are the
output gates.  Note that the circuit forms a directed acyclic graph (DAG), when gates are considered as nodes, and there is
an edge from $g_j$ and $g_k$ to $g_i$ if $g_i=g_j \circ g_k$. 
Since, circuit $C_I$ represents function $F_I$ it has to be the case that if we input $\ll \in [0,\ 1]^d$ to $C_I$ then all
the gates are well defined and the circuit outputs $C_I(\ll)=F_I(\ll)$ in $[0,\ 1]^d$.
We note that a circuit representing a problem in FIXP operates on real numbers, but 
the underlying model of computation is still the standard discrete Turing machine. 
In other words, an algorithm for FIXP problems is not allowed to do any computation on reals.

Let $*\zeta$ denote multiplication by a rational constant $\zeta\in \QQ$.  The Linear-FIXP is a subclass of FIXP where the
operations are restricted to $\circ \in \{max, +, *\zeta\}$.  A function defined by a Linear-FIXP circuit is polynomial
piecewise-linear, and all its fixed points are rational numbers of size $poly(L)$ \cite{EY07}, where $L$ is the total size
of the circuit which is \#inputs + \#gates + total size of the constants used in the circuit.  Etessami and Yannakakis
showed that PPAD = Linear-FIXP.  Next, we define a subclass of Linear-FIXP based on the number of inputs and outputs.

\begin{definition}
For a $k \ge 1$, an instance $I$ is in $kD$-Linear-FIXP if  
$F_I:[0,\ 1]^k \rightarrow [0,\ 1]^k$. i.e., $F_I$ is defined by a circuit with $k$ inputs and $k$ outputs. 
\end{definition}

Since fixed-point of a $1$-dimensional piecewise-linear function can be computed in polynomial time using a binary search, 
$kD$-Linear-FIXP is in $P$ for $k=1$. 
But for any constant $k>1$ it is not clear if the problem is in $P$ or it is hard.
In the next section, we show that the problem is PPAD-hard even for $k=2$.

\section{PPAD-hardness of 2D-Linear-FIXP}\label{sec.2dlf}
In this section we describe the construction of a Linear-FIXP circuit with two inputs and two outputs, from an instance of $2D$-Brouwer defined
by a Boolean Brouwer-mapping circuit. We show that the function defined by the resulting $2D$-Linear-FIXP circuit is such that all its
fixed-points are in trichromatic squares of the $2D$-Brouwer instance, 
thereby proving PPAD-hardness of $2D$-Liner-FIXP using \cite{CD2D}.

Let $C^b$ be the valid Brouwer-mapping circuit of a given 2D-Brouwer instance on grid $G_{n}$, and $H$ be the discrete function
defined by circuit $C^b$.  We construct a Linear-FIXP
circuit $C$ that computes a function $F:[0,\ 2^n-1]^2\rightarrow [0,\ 2^n-1]^2$, an extension of the discrete
function $H$. 

Recall that given a bit representation of a grid point $\pp \in G_{n}$, circuit $C^b$ outputs four bits $\D^+_1, \D^-_1,\D^+_2,\D^-_2$,
so that for $I=(\dpo-\dmo, \dpt - \dmt)$, $H(\pp)=\pp+I$. Similarly, for every non-grid point $\pp =(p_1,p_2) \in K_{\qq}$, we need to
compute an incremental vector based on the incremental vectors of the vertices of $K_{\qq}$. For this we need to extract the
integer parts of $p_1$ and $p_2$, i.e., compute $\lfloor p_1\rfloor$ and $\lfloor p_2\rfloor$, and then its bit representation. 
Since, {\em floor} is a discontinuous function, it can not be computed using Linear-FIXP operations, which are inherently continuous.
However, next we achieve this 
for the points not very near to the boundary of any cell. 

Recall that the operations allowed in a Linear-FIXP circuit are $\{max, +, *\zeta\}$. Clearly, $\{min, -\}$ can be
simulated using the allowed operations. Let $L >16 $ be a large integer with value being a power of $2$, and at most polynomial in
$size[C^b]$, {\em i.e.}, $L=2^l \le poly(size[C^b])$.
Consider the {\em ExtractBits} procedure of Table \ref{tab1}.

\begin{table}[!h]
\begin{center}
\begin{tabular}{|l|}\hline
ExtractBits(a)\\
$x\leftarrow a$ \\ 
{\bf for} i=n-1 {\bf to} 0 {\bf do} \\
\hspace{10pt} $b_{i} \leftarrow \min\{\max\{((x-2^i)*L^2)+1,0\},1\}$\\
\hspace{10pt} $x\leftarrow x-2^ib_{i}$\\
{\bf endfor}\\
Output the bit vector $\bb=(b_{n-1},\dots, b_0)$.\\\hline
\end{tabular}
\caption{Extract Bits of the Integer Part}
\label{tab1}
\end{center}
\vspace{-0.5cm}
\end{table}

\begin{definition}\label{def.well-poor}
We say that $a \in \Real_+$ is {\em poorly positioned} if for some integer $t \in \ZZ_+$, $a=t+\epsilon$, where
$1-\frac{1}{L^2} < \epsilon < 1$. A point $\pp \in \Real_+^2$ is said to be {\em poorly-positioned}, if any of its
coordinates is poorly positioned, otherwise it is called {\em well-positioned}.
\end{definition}

\begin{lemma}\label{lem.be}
Given a well-positioned number $a \in [0,\ 2^n)$, 
the vector $\bb=ExtractBits(a)$ is a bit representation of $\lfloor a\rfloor$.
\end{lemma}
\begin{proof}
Let $a=a'+\epsilon$, where $a' \in \ZZ_+$ and $0 \le \epsilon \le 1-\frac{1}{L^2}$.  We show that every $b_i$ is either $0$
or $1$, and is set correctly.  Proof is by induction. If $a'\ge 2^{n-1}$, then clearly, $(a-2^{n-1})*L^2+1 \ge 1$ and
$b_{n-1}$ will be one. If $a'< 2^{n-1}$, then $(a-2^{n-1})*L^2+1 \le (-1+\epsilon)*L^2 +1 \le (-1 +1 -\frac{1}{L^2})L^2 +1
\le 0$ and hence $b_{n-1}$ will be zero. In either case $x = a - b_{n-1} 2^{n-1}$ will satisfy the hypothesis, and we can
apply the same argument for bit $b_{n-2}$.
\end{proof}

Given a well positioned point $\pp \in K_\qq$, we can extract bit representations of each of the coordinates of $\qq$ 
due to Lemma \ref{lem.be}, and hence of all the vertices of $K_\qq$. Next task is to obtain each of their incremental
vectors by simulating circuit $C^b$ in Linear-FIXP. 
Circuit $C^b$ is a Boolean circuit with operations $\land,\lor$ and $\lnot$ and takes only Boolean input. 
These operations are easy to simulate in Linear-FIXP:
If $a, b \in \{0, 1\}$, then clearly $a \land b = \min\{a,b\}$, $a\lor b=\max\{a,b\}$ and $\lnot a = (1-a)$.  

Thus, if $\pp$ is well positioned, then incremental vectors of the vertices of $K_\qq$ can be computed using a Linear-FIXP
circuit. However, if $\pp$ is poorly-positioned, then Lemma \ref{lem.be} provides no guarantees and indeed the {\em
ExtractBits} procedure may produce vector $\bb$ with the value $b_i$s being anything in $[0,\ 1]$.
This is expected due to continuity property of Linear-FIXP operations.
Similar difficulty arises in the approaches of Daskalakis et al. \cite{DGP} and Chen et al. \cite{CDT}. Both resort to a
sampling argument, first proposed in \cite{DGP}, and later improved in \cite{CDT}. Next we describe a version of \cite{CDT}
argument.

Given a set of points $S=\{\pp^1,\dots, \pp^l\}$, let $I_w(S)$ and $I_p(S)$ denote the set of indices of the well and poorly
positioned points of $S$ respectively. Given $\pp \in \R^2_+$, let $\pi(\pp)
= \{\qq \ |\ q_1, q_2 \mbox{ are the largest}$ $\mbox{integers from } \{0,\dots,2^n-1\} \mbox{ s.t. } q_1$ $\le p_1 \mbox{ and } q_2\le
p_2\}$.  For $\ee^1=(1,0), \ee^2=(0,1)$ and $\ee^0=(-1,-1)$, let
$\zeta(\pp)=\ee^i$, where $i={g_{C^b}(\pi(\pp))}$.

\begin{lemma}\label{lem.sampling}
Given $\pp \in [0,\ 2^n-1]^2$, consider the set $S=\{\pp^1,\dots,\pp^{16}\}$ such that
\[
\pp^j = \pp + (j-1)(\frac{1}{L}, \frac{1}{L}),\ \ \ \ j \in [16]
\]

For each $j \in I_p(S)$, let $\rr^j \in \R^2$ be a vector with $\|\rr^j\|_\infty \le 1$.
And for each $j \in I_w(S)$, let $\rr^j=\zeta(\pp^j)$.
If $\| \sum_{j=1}^{16} \rr^j\|_\infty =0$ then $K_{\pi(\pp)}$ is trichromatic.
\end{lemma}
\begin{proof}
Let $Q=\{\qq^j=\pi(\pp^j) \ |\ \pp^j \in S\}$. Since $\frac{16}{L} <<1$ the set crosses boundaries of cells at most twice.
In other words, for each $i=1,2$, there is at most one $j_i$ such
that $q_i^{j_i} = q_i^{j_i-1}+1$. Therefore, set $Q$ can have at most three elements, and they are part of the same square
which has to be $K_{\pi(\pp)}$. 

Further, since $\frac{1}{L^2}<<\frac{1}{L}<<1$, there can be at most two poorly-positioned points in $S$. So, we have
$|I_w(S)|\ge 14$.  Let $\rr^G = \sum_{j\in I_w(S)} \rr^j$, then we have
$\|\rr^G + \sum_{j \in I_p(S)} \rr^j\|_\infty = 0 \Rightarrow \|\rr^G\|_\infty \le \|\sum_{j \in I_p(S)} \rr^j\|_\infty \le 2$, 
because $|I_p(S)|\le 2$ and $\|r^j\|_\infty \le 1$ for each $k \in I_p(S)$.

Let $W_i$ be the number of indices of $I_w(S)$ with $\rr^j=\ee^i$. Using the above fact, we will show that $W_i\neq
0, i=0,1,2$, to prove the lemma.  

If $W_0=0$ then $W_i \ge 7$ for either $i=1$ or $i=2$. In that case, $r^G_i\ge 7$, a contradiction. If $W_t=0$
for $t=1 \mbox{ or } 2$,
then $W_0 < 3$ or else $r^G_t \ge 3$. Let $i^* = \argmax_{0\le i\le 2} W_i$, then clearly $W_{i^*} \ge 7$ and $i^*\neq 0$. Then,
$r^G_{i^*} \ge 7-2 = 5$, again a contradiction. 
\end{proof}

\begin{remark}
Note that, in Lemma \ref{lem.sampling}, it suffices to assume $\| \sum_{j=1}^{16} \rr^j\|_\infty < 1$ for $\pp$ to be in a trichromatic
square. We use this fact to derive inapproximability results in Section \ref{sec.approx}.
\end{remark}

Lemma \ref{lem.sampling} implies that even if point $\pp$ is poorly positioned, we can make sure that it forms a
fixed-point only when it is in a trichromatic square by sampling $16$ carefully chosen points near it. 
Next we describe a complete construction of the Linear-FIXP circuit $C$, and then show its correctness using Lemmas
\ref{lem.be} and \ref{lem.sampling}.

\begin{itemize}
\item[$S_1$.] Let $p_1$ and $p_2$ denote the two inputs of the Linear-FIXP circuit. These are any real number from $[0,\ 2^n-1]$.
Compute 16 points using the $+$ gates and rational constants: 
\[
\pp^i=\pp+(j-1)(\frac{1}{L},\frac{1}{L}),\ \ \ \ j \in [16]
\]
\item[$S_2$.] Call ExtractBits($\pp^j_t$), $t=1,2$ and $j \in [16]$, and let the output vector be $\bb^{j,t}$. 
\item[$S_3$.] For each $j \le [16]$, feed $b^{j,1}_0, \dots, b^{j,1}_{n-1},b^{j,2}_{0},\dots b^{j,2}_{n-1}$ to a 
simulation of circuit $C^b$, where $\lor$, $\land$ and $\lnot{x}$ are replaced with $\max$, $\min$ and $1-x$ respectively.
Note that, there are total of 16 simulations of circuit $C^b$.
Let $\Djpo, \Djmo,\Djpt,\Djmt$ be the output values of these.
\item[$S_4$.] For each $j \in [16]$, compute $r^j_1=\min\{\max\{\Djpo-\Djmo, -1\},1\}$ and $r^j_2=\min\{\max\{\Djpt-\Djmt,
-1\},1\}$.
\item[$S_5$.] Compute $r_1=\frac{1}{16}\sum_{j \in [16]} r^j_1$, and $r_2=\frac{1}{16}\sum_{j \in [16]} r^j_2$. 
\item[$S_6$.] Output $p'_1=\max\{\min\{p_1+r_1,2^n-1\},0\}$ and $p'_2=\max\{\min\{p_2+r_2,2^n-1\},0\}$.
\end{itemize}

The number of gates used in steps $S_1$, $S_4$, $S_5$ and $S_6$ of the above procedure are constant. We used $O(n)$ gates in
step $S_2$, and 16 times as many as the number of gates in $C^b$ in step $S_3$. Further, since value of $L$ is polynomial in $size[C^b]$,
the constants used in steps $S_1$, $S_2$ and $S_5$ are polynomial sized. Thus, the total size of the Linear-FIXP circuit $C$
constructed by the above procedure is polynomial in $size[C^b]$. Next we show that each of the fixed-points of function $F$
represented by circuit $C$ are in trichromatic squares of the grid $G_{n}$.

\begin{lemma}\label{lem.red1}
Every fixed point of $F$ is inside a trichromatic square of $G_n$.
\end{lemma}
\begin{proof}
Let $\pp \in [0,\ 2^n-1]^2$ be a fixed point of $F$.  If $\pp \in (0,2^n-1)^2$ then for it to be a fixed point, the final
incremental vector $\rr$ has to be $(0,0)$.  Let $S=\{\pp^j\ |\ j \in [16]\}$. Due to Lemma \ref{lem.be}, we know that for
each $j \in I_w(S)$ we have $\rr^j=\zeta(\pp^j)$.  Further, due to step ($S_4$) for each $k \in I_p(S)$, $\|\rr^j\|_\infty
\le 1$.  Therefore, using the fact that $\rr=\frac{1}{16}\sum_{j=1}^{16} \rr^j$ and Lemma \ref{lem.sampling} it follows that
$K_{\pi(\pp)}$ is trichromatic.

For the remaining case, $\pp$ has to be on a boundary of the grid. Since $C^b$ is a valid circuit, vertices on the boundary
has specific incremental vectors: Let $\qq$ be such a vertex then if $q_2=0$ then $\zeta(\qq)=\ee^2=(0,1)$, else if $q_1=0$
then $\zeta(\qq)=\ee^1=(1,0)$, otherwise $\zeta(\qq)=\ee^0=(-1,-1)$.  Using this fact, and that $|I_w(S)|\ge 14$ (Lemma
\ref{lem.sampling}), next we show $\pp$ can not be a fixed point in that case.

If $p_2=0$, then for each $k \in I_w(S)$, $\rr^j=(0,1)$. Therefore, we have $r_2>0$ and in turn $p'_2>p_2$.
If $p_2>0$ and $p_1=2^n-1$, then for each $k \in I_w(S)$, $\rr^j$ is either $(0,1)$ or $(-1,-1)$, and one of them occurs at
least $7$ times. Therefore, either $r_2>0$ and in turn $p'_2>p_2$, or $r_1<0$ and in turn $p'_1<p_1$.

If $0<p_1<2^n-1$ and $p_2=2^n-1$, then for each $j \in I_w(S)$, $\rr^j$ is either $(1,0)$ or $(-1,-1)$. Therefore, we have
either $r_1>0$ and in turn $p'_1>p_1$, or $r_2<0$ and in turn $p'_2<p_2$.
If $p_1=0$ and $1\le p_2<2^n-1$, then for each $j \in I_w(S)$, $(r^j_1,r^j_2)=(1,0)$. Therefore, we have $r_1>0$ and in turn
$p'_1>p_1$. Further, if $p_1=0$ and $0 < p_2 < 1$, then by similar argument either $p'_1>p_1$ or $p'_2>p_2$.
\end{proof}

\begin{remark}
Note that every fixed-point of $F$ is in a trichromatic square whose vertices with the three colors form a right-angled triangle with
north-east oriented hypotenuse.
\end{remark}

It is easy to shrink the range of $F$ from $[0,\ 2^n-1]$ to $[0,\ 1]$. Consider a function $F':[0,\ 1]^2\rightarrow
[0,\ 1]^2$, such that $F'(\l_1,\l_2)=\frac{1}{2^n-1}F((2^n-1)\l_1,(2^n-1)\l_2)$, then clearly, $(\l_1,\l_2)$ is a
fixed-point of $F'$ if and only if $((2^n-1)\l_1,(2^n-1)\l_2))$ is a fixed-point of $F$. Thus we get the following
theorem using Lemma \ref{lem.red1} and the fact that $size[C]=poly(size[C^b])$.

\begin{theorem}\label{thm.kdlf}
The class of $kD$-Linear-FIXP with $k>1$ is PPAD-hard.
\end{theorem}

\begin{remark}
All the known PPAD-hardness proofs for games, go through {\em generalized circuits} \cite{DGP,CDT}, which allows feedback-loops and
approximate computation for each operation.  
However, in \cite{DGP} and \cite{CDT} feedback-loops  
are used only to connect the "output nodes" to the "input nodes" to ensure that their values are almost same (approximate solution).
Further, each operation of generalized circuit can be simulated using $\{max, +, *\zeta\}$, and polynomial sized rational
numbers, and most of them with exact computation. The reduction discussed in this section can be obtained using these observations as
well from the previous approaches of reducing $2D$ or $3D$-Brouwer to generalized circuit.
\end{remark}

\section{Reduction: kD-Linear-FIXP to Rank-(k+1) Game}\label{sec.r3}
Given a $kD$-Linear-FIXP instance, with function $F:[0,\ 1]^k$ $\rightarrow [0,\ 1]^k$ represented by circuit $C$, in this
section we construct a rank-$(k+1)$ bimatrix game whose Nash equilibria are almost\footnote{Essentially, Nash equilibrium
strategies of the first player are in one-to-one correspondence with the fixed points} in one-to-one correspondence with the
fixed points of $F$. We do this in two steps. First we replace the circuit $C$, by a parametric linear program (LP) with
$k$-parameters, where inputs of circuit $C$ become parameters of the LP.  Given values of the $k$ inputs, we show that 
the $k$ outputs of the circuit $C$ are linear function of a solution of the LP. This defines a function $F^{lp}$ from $\R^k$
to $\R^k$, and we show that the
fixed points of $F^{lp}$ are in one to one correspondence with the fixed-points of $F$.  Later, we construct a rank-$(k+1)$
game using the LP and its dual, whose Nash equilibria exactly captures the fixed points of $F^{lp}$. 

\begin{remark}
Recall that linear programs are equivalent to zero-sum games \cite{dantzig,adler_zerosum}. However, the reductions from LP to
zero-sum games constructs a symmetric game, and require to compute a symmetric Nash equilibrium. There are no such
restrictions in our construction, however our reduction is not general enough and uses the fact that the parametric LP has
been 
constructed from a Linear-FIXP circuit. It will be interesting to reduce an LP to a non-symmetric zero-sum game, and also a
fixed-point problem with
parametric LP to a constant rank game in general.
\end{remark}

\subsection{Replacing Linear-FIXP circuit with a linear program}\label{sec.lp}
In this section we construct a parameterized linear program with $k$ parameters, which can replace a $kD$-Linear-FIXP circuit.
Let $C$ be a $kD$-Linear-FIXP circuit representing the function $F:[0,\ 1]^k\rightarrow [0,\ 1]^k$. 
Circuit $C$ being a Linear-FIXP circuit, it allows only three operations, namely
$\max, +$ and $*\zeta$ where $\zeta$ is a rational number, and it forms a DAG. 
The $size[C]$ is \# inputs +\# gates +\# total bit lengths of the constants
in the circuit. 

If $C$ is considered as a function from $\R^k$ to $\R^k$, then it is same as function $F$ on $[0,\ 1]^k$, but can be
anything outside this range and hence may have fixed-points outside $[0,\ 1]^k$ as well. 
To prevent this, we add two $\max$ gates for every output of the circuit, as follows:
Let $\tau_1,\dots,\tau_k$ be the $k$ outputs of circuit $C$. 
Without loss of generality (wlog), we will add two $\max$ gates for each $l \in [k]$ to ensure that each output value is in $[0,\ 1]$:
\begin{equation}\label{eq.m}
\max\{0,\min\{1,\tau_l\}\} = \max\{0,-1*\max\{-1,-1*\tau_l\}\} 
\end{equation}

The above transformation ensures that the output vector of $C$ is always in $[0,\ 1]^k$, 
and hence fixed-points of $C$ are
exactly the fixed-points of $F$. 
Next, we show that it is wlog to assume that one of the inputs of every $\max$ gate is zero. 

\begin{lemma}\label{lem.max0}
Given a circuit $C$, it can be transformed to an equivalent polynomial sized circuit where one of the inputs of every $\max$
gate is zero.  
\end{lemma}
\begin{proof}
Consider a $\max$ gate, and let $a$ and $b$ be the inputs and $c$ be the output, then we
have $c=\max\{a,b\}$ which is equivalent to $c=\max\{0,b-a\}+a$. Therefore, we can transform circuit $C$ such that one input of every
max gate is $0$. This transformation requires $3$ extra gates per $\max$ gate, two $+$ and one $*\zeta$ where $\zeta=-1$. 
Clearly, the increase in the size of the circuit is polynomial.
\end{proof}

Wlog we assume that every $\max$ gate of circuit $C$ has exactly one non-trivial input, and the other input is
always zero (due to Lemma \ref{lem.max0}). Let $m$ be the number of $\max$ gates in $C$. Since $C$ is a DAG, there is an
ordering among the $\max$ gates, say $g_1,\dots,g_m$, such that if there is a path from $g_i$ to $g_j$ in $C$ then $i<j$;
ties are broken arbitrarily. 
Let $n=m-2k$ be the number of max gate in the original circuit, before the addition of (\ref{eq.m}) per output.
Let these be the first $n$ max gates.
Let the ordering be such that these are the first $n$ gates $g_1,\dots,g_n$. 
In (\ref{eq.m}) let $g_{n+2l-1}$ denote the inner $\max$ gate and $g_{n+2l}$ denote the outer one, then $k$
outputs of circuit are the outputs of gates $g_{n+2l},\  l \in [k]$.

Let the $k$ inputs of circuit $C$ be denoted by $\ll=(\l_1,\dots\l_k)$, and let $x_i$ capture the output of the $i^{th}$
$\max$ gate. 
Note that, $x_{n+2l},\ l \in [k]$ are the output of the circuit.
Except for the $\max$, rest of the two operations give rise to linear expressions in the
$\ll$ and $x_i$s of the previous $\max$ gates. We use this observation crucially in the rest of the construction. 

Note that, for each $i \in [m]$, the input of $g_i$ is a linear expression in
$x_1,\dots,x_{i-1},\l_1,\dots,\l_k$, with a constant term. We denote this expression by $L_i(x_1,\dots,x_{i-1},\ll)$, then
the following conditions exactly capture the operation of $g_i$.
\begin{equation}\label{eq.m1}
\forall i \in [m],\ \ \ x_i \ge 0,\ \ \ \ \ \ \ x_i \ge L_i(x_1,\dots,x_{i-1},\ll)
\end{equation}
\begin{equation}\label{eq.m2}
\forall i \in [m],\ \ \ x_i (x_i - L_i(x_1,\dots,x_{i-1},\ll))=0
\end{equation}

The next lemma follows by construction.

\begin{lemma}\label{lem.max}
Given $\ll \in \R^k$, $(\xx,\ll)$ satisfies (\ref{eq.m1}) and (\ref{eq.m2}) iff when $\ll$ is given as the input to circuit
$C$, the $i^{th}$ $\max$ gate
evaluates to $x_i$ for all $i \in[m]$. 
\end{lemma}
\begin{proof}
Reverse direction follows just by construction.
For the forward direction we will argue by induction. Suppose, $(\xx,\ll)$ satisfies (\ref{eq.m1}) and (\ref{eq.m2}). Then,
for $\ll$ as input to $C$, clearly $L_1$ evaluates to exactly the input of the (first $\max$) gate $g_1$. In that case,
(\ref{eq.m1}) forces that $x_1$ is at least as large as inputs of $g_1$, and (\ref{eq.m2}) forces that it equals one of the
input. Thus, $x_1$ captures output of $g_1$. Now, suppose this is true for first $k\ge 1$ $\max$ gates. Then for
$(k+1)^{th}$ $\max$ gate, again $L_{k+1}$ is exactly the input of $g_{k+1}$, and the lemma follows by the same argument.
\end{proof}

Constraints of (\ref{eq.m1}) gives a system of linear inequalities,
\begin{equation}\label{eq.1}
A\xx \ge \sum_{l=1}^k {\l_l \uu^l} + \bb,\ \ \ \ \ \xx \ge 0
\end{equation}
where, $\bb$ and $\uu^l,\ l \in [k]$ are $m$-dimensional rational vectors, and $A$ is an $m \times m$ lower-triangular
rational matrix with ones on the diagonal.
Once we plugin some values for $\l_1,\dots,\l_k$, (\ref{eq.1}) becomes a polyhedron in $\xx$. Let it be denoted by $\CP(\ll)$.
For any $\ll \in \R^k$ and $\xx \in \CP(\ll)$, vector $(\xx,\ll)$, satisfies (\ref{eq.m1}).

\begin{remark}\label{rem.1}
Enforcing (\ref{eq.m2}) requires quadratic complementarity-type constraints. Using this fact, in Section \ref{sec.simple} we give a
simplified proof of PPAD-harness of $2$-Nash and symmetric $2$-Nash (bypassing even the parameterized LPs). 
\end{remark}

Next, we construct a cost vector $\cc \in \R^m$, such that 
minimizing $\xx^T\cdot \cc$ over $\CP(\ll)$ will give a solution that, together with $\ll$, satisfies
(\ref{eq.m2}) as well.

\begin{table}[!h]
\begin{center}
\begin{tabular}{|l|}\hline
ConstructCost($A$)\\
$c_m \leftarrow 1$, $\b_m\leftarrow 1$\\
{\bf for} $i=m-1$ to $1$ {\bf do}\\
\hspace{15pt}$c_i\leftarrow \sum_{j>i} {|a_{ji}| \b_j} +1$, $\ \ \ \b_i \leftarrow c_i +  \sum_{j>i} {|a_{ji}| \b_j}$\\
{\bf endfor};\\
Output $\cc$\\\hline
\end{tabular}
\caption{Construction of the Cost Vector}\label{tab2}
\end{center}
\vspace{-0.4cm}
\end{table}

For $\cc=$ConstructCost($A$), consider the following parameterized LP and its dual.

\begin{equation}\label{eq.lp}
\begin{array}{ll}
LP(\ll): 
\begin{array}{l}
\min:\ \cc^T \cdot \xx \\
s.t.,\ \ \ A\xx \ge \sum_{l\in [k]} {\l_l \uu^l} + \bb \\ 
\ \ \ \ \ \ \ \ \ \xx \ge 0
\end{array}
&
\hspace{1.5cm}
DLP(\ll): 
\begin{array}{l}
\max:\ (\sum_{l\in [k]} {\l_l \uu^l} + \bb)^T\cdot \yy \\
s.t.,\ \ \ A^T\yy \le \cc \\
\ \ \ \ \ \ \ \ \ \yy \ge 0
\end{array}
\end{array}
\end{equation}

The complementary slackness requires that
solutions of $LP(\ll)$ and $DLP(\ll)$ satisfy (KKT conditions),
\begin{equation}\label{eq.cs}
\forall i \in [m],\ y_i(A\xx-\sum_{l\in [k]} \l_l\uu^l -\bb)_i=0,\ x_i(A^T\yy-\cc)_i=0
\end{equation}

\begin{lemma}\label{lem.xy}
Given $\ll \in \R^k$, $\xx$ is a solution of $LP(\ll)$ iff $(\xx,\ll)$ satisfies (\ref{eq.m1}) and (\ref{eq.m2}).
\end{lemma}
\begin{proof}$(\Rightarrow)$
Let $\yy$ be the dual solution corresponding to $\xx$, i.e., $(\xx,\yy)$ satisfies (\ref{eq.cs}).
Since $\xx$ is a feasible point of $LP(\ll)$, clearly, $(\xx,\ll)$ satisfies (\ref{eq.m1}). For (\ref{eq.m2}),
it suffices to show that $\forall i \in [m], x_i>0 \Rightarrow y_i>0$, then the proof follows using (\ref{eq.cs}).

Let $\bbeta \in \R^m$ be the vector calculated in ConstructCost($A$) of Table \ref{tab2}.
We do the proof by induction, where we show that $\forall i \in [m]$, $y_i \le \b_i$, and $x_i >0 \Rightarrow y_i>0$.
Recall that $A$ is lower-triangular with ones on the diagonal. Therefore, $A^T$ is upper-triangular with ones on the diagonal.

Our base case is when $i=m$: If $x_m>0$, then due to (\ref{eq.cs}) we have $y_m=(A^T\yy)_m=c_m=1>0$. Further, 
$(A^T\yy)_m \le c_m \Rightarrow y_m \le 1$. Since $\b_m=1$ we get $y_m \le \b_m$. 

Now, let the hypothesis be true for $j > r$. 
For $r$ if $x_r > 0 $ then $(A^T\yy)_r=c_r\Rightarrow (A^T\yy)_r=y_r + \sum_{j>r}
a_{jr} y_j = c_r$ (due to (\ref{eq.cs})). Since, $\forall j > r, 0 \le y_j \le \b_j$ and $c_r = \sum_{j>r} |a_{jr}|\b_{j} + 1$, we
have $\sum_{j>r} a_{jr}y_j < c_r$.  Therefore, for the equality to hold we must have $y_r>0$. Further, $(A^T\yy)_r = y_r + \sum_{j>r}
a_{jr} y_j \le c_r \Rightarrow y_r \le c_r - \sum_{j>r} a_{jr} y_j \le c_r + \sum_{j>r} |a_{jr}|y_j \le c_r + \sum_{j>r} |a_{jr}|\b_j =
\b_r$. 

$(\Leftarrow)$ If $(\xx,\ll)$ satisfies (\ref{eq.m1}) and (\ref{eq.m2}) then clearly $\xx$ is feasible in $LP(\ll)$. 
Construct $\yy$, from $y_m$ to $y_1$ as follows: if $x_r=0$ then set $y_r=0$, else set $y_r=c_r - \sum_{j > r} a_{jr} y_j$.
It is easy to see that $\yy$ is feasible in $DLP(\ll)$, 
and it, together with $(\xx,\ll)$ satisfies (\ref{eq.cs}).
\end{proof}

Lemmas \ref{lem.max} and \ref{lem.xy} imply that $LP(\ll)$ simulates the circuit. Next, we show that the circuit can be replaced by
$LP(\ll)$ without affecting the fixed-points of $F$.
Consider function $F^{lp}:\R^k \rightarrow [0,\ 1]^k$, such that, 
\begin{equation}\label{eq.flp}
\ll \in \R^k,\ \ \ \ F^{lp}(\ll) = (x_{n+2l})_{l\in [k]},\ \  \mbox{ where } \xx=LP(\ll)
\end{equation}

We show that function $F^{lp}$ is well-defined, and its fixed-points are exactly the fixed-points of $F$. 

\begin{lemma}\label{lem.flp}
$F^{lp}$ is well defined, and $\ll \in \R^k$ is a fixed-point of $F^{lp}$ iff it is a fixed point of $F$.
\end{lemma}
\begin{proof}
For any given $\ll \in \R^k$, using Lemmas \ref{lem.max} and \ref{lem.xy}, it follows that $LP(\ll)$ has a unique solution, and if it
is $\xx$ then $0\le x_{n+2l}\le 1,\ \forall l \in[k]$. Thus, $F^{lp}$ is well defined.

For the second part, we know that $F^{lp}(\ll)\in [0,\ 1]^k$. Since, $\ll$ is a fixed point, this also forces $\ll \in [0,\
1]^k$. Further, since circuit $C$ represents the function $F$ in range $[0,\ 1]^k$, it suffices to show that $F^{lp}(\ll)=C(\ll)$.

In other words, when vector $\ll$ is the input to circuit $C$, then $i^{th}$ $\max$ gate evaluates to $x_i$,
$\forall i \in [m]$, where $\xx=LP(\ll)$. This follows using Lemmas \ref{lem.max} and \ref{lem.xy}. 
\end{proof}

\begin{lemma}\label{lem.size}
Size of matrix $A$, and vectors $\cc$, $\bb$ and $\uu^l,\ \forall l \in [k]$ are polynomial in $size[C]$.
\end{lemma}
\begin{proof}
By construction $A$, $\bb$ and $\uu^l,\ \forall l \in [k]$, are formed by the coefficients of the linear expressions $L_i$s of
(\ref{eq.m1}). These linear expressions are constructed due to the $+$ and $*\zeta$ gates of the circuit $C$, therefore, the absolute
value of any of its co-efficient is at most $\zeta_{max}^v$, where $v$ is the number of $*\zeta$ gates in $C$, and $\zeta_{max}$ is the maximum
absolute rational constant used in $C$. For rational constants $\zeta_1,\zeta_2$, since
$size(\zeta_1*\zeta_2) = size(\zeta_1)+size(\zeta_2)$, we have that the size of every co-efficient of $L_i$s is at most $size[C]$.
Thus, sizes of $A$, $\bb$ and $\uu^l,\ \forall l\in [k]$ are at most polynomial in $size[C]$.  
Let $A_{max}=\max_{i,j \in [m]} A_{ij}$, then by construction $c_1=\max_{j\in [m]} c_j \le (2A_{max}+1)^n$ (see Table \ref{tab2}).
Therefore, the size of $\cc$ is also bounded by a polynomial in $size[C]$.
\end{proof}

From Lemmas \ref{lem.flp} and \ref{lem.size} we can conclude that finding a fixed point of $F$ is equivalent to finding one for
function $F^{lp}$, which can be represented using polynomially many bits in the $size[C]$. 
Next we reduce the fixed-point computation in $F^{lp}$ to Nash equilibrium computation in a rank-$(k+1)$ game, such that the size of
the game is polynomial in size of the parameters of function $F^{lp}$.

\subsection{Constructing Rank-(k+1) Game}\label{sec.lcp}
Since, feasibility and complementary slackness are necessary and sufficient conditions for the solutions of an LP, it is well known
that an LP can be formulated as a linear complementarity problem (LCP). Using this, next we construct an LCP whose solutions are exactly the fixed
points of $F^{lp}$. Before we do this, note that since all the co-ordinates of the cost vector $\cc$ are strictly positive,
we can make it all ones vector by dividing $j^{th}$ column of $A$ by $c_j$.  Let $H$ be the transformed matrix, {\em i.e.,}
\[H_{ij}=\nfrac{A_{ij}}{c_j}\] 

To reflect this transformation in LPs of (\ref{eq.lp}) define,

\begin{equation}\label{eq.tlp}
\begin{array}{ll}
LP'(\ll): 
\begin{array}{l}
\min:\ \sum_i x_i \\
s.t.,\ \ \ H\xx \ge \sum_{l\in [k]} {\l_l \uu^l} + \bb \\ 
\ \ \ \ \ \ \ \ \ \xx \ge 0
\end{array}
&
\hspace{1.5cm}
DLP'(\ll): 
\begin{array}{l}
\max:\ (\sum_{l\in [k]} {\l_l \uu^l} + \bb)^T\cdot \yy \\
s.t.,\ \ \ H^T\yy \le \ones \\
\ \ \ \ \ \ \ \ \ \yy \ge 0
\end{array}
\end{array}
\end{equation}

The next lemma follows by construction.

\begin{lemma}\label{lem.equiv}
Given $\ll \in \R^k$, $\xx$ and $\yy$ are solutions of $LP(\ll)$ and $DLP(\ll)$ respectively iff, $\xx'$, where 
$x'_j=x_jc_j,$ $\forall j \in [m]$, and $\yy$ are solutions of $LP'(\ll)$ and $DLP'(\ll)$ respectively. 
\end{lemma}

For $l \in [k]$, let $\vv^l$ be an $m$-dimensional vector with $n+2l^{th}$ co-ordinate set to $\nfrac{1}{c_{n+2l}}$ and
rest all set to zero, then $\vv^{l^T}\cdot \xx' = \nfrac{x'_{n+2l}}{c_{n+2l}}$. 
Lemma \ref{lem.equiv} implies that, 
at a fixed-point of $F^{lp}$ we have $\l_l=x_{n+2l}=\nfrac{x'_{n+2l}}{c_{n+2l}}=\vv^{l^T}\cdot \xx',\ \forall l \in [k]$, where
$\xx=LP(\ll)$ and $\xx'=LP'(\ll)$. Using this as a motivation, 
we replace $\l_l$ with $(\vv^{l^T}\cdot \xx)$ in the constraints of $LP'(\ll)$. The resulting matrix
will be \[H'=H-\sum_{l=1}^k \uu^l\cdot \vv^{l^T},\ \forall (i,j)\]

Using the above observation, and feasibility and complementary slackness conditions for (\ref{eq.tlp}),
we construct the following LCP, called $\lcp_C$,
\begin{equation}\label{eq.lcp}
\begin{array}{cll}
& H'\xx \ge \bb;\ \  & H^T\yy \le \ones \\
& \xx \ge \zeros;\ \  & \yy \ge \zeros \\
\forall i \in [m],& y_i(H'\xx -b)_i=0; \ \ & x_i(H^T\yy- \ones)_i=0
\end{array}
\end{equation}

Before we connect the solutions of $\lcp_C$ with the fixed-points of $F^{lp}$, we need to establish a few properties
about $H$, $H'$, $\bb$ and $\uu^l$s. For this we need to understand (\ref{eq.m1}) for the last $2k$ $\max$ gates that we
added in (\ref{eq.m}) to ensure that the outcome of the circuit is in $[0,\ 1]^k$.  Due to Lemma \ref{lem.max0} we have
assumed that one of the inputs of every $\max$ gate is zero. For this to be the case, the (\ref{eq.m}) has to be transformed
as follows,
\[
\forall l \in [k],\ \ \ \max\{0,-1*(\max\{0,-\tau_l+1\}-1)\}
\]

Here, $\forall l\in [k], \ \tau_l$ is a linear expression in $x_1,\dots,x_n,\ll$, and in turn so is $L_{n+2l-1}=1-\tau_l$.
Recall that $x_{n+2l-1}$ captures the output of the inner $\max$ gate and $x_{n+2l}$ captures the output of the outer $\max$
gate. Therefore, we have
\begin{equation}\label{eq.m3}
\begin{array}{lll}
\forall l \in [k],& x_{n+2l-1}\ge0, \ & x_{n+2l-1} \ge L_{n+2l-1}(x_1,\dots,x_n,\ll)\\
\forall l \in [k],& x_{n+2l}\ge 0,\ & x_{n+2l} \ge 1-x_{n+2l-1} \Rightarrow x_{n+2l-1}+x_{n+2l}\ge 1
\end{array}
\end{equation}

The following properties are easy to obtain using (\ref{eq.m3}). 
\begin{itemize}
\item[$P_1$.] $\forall l\in [k]$, $(A\xx)_{n+2l}=x_{n+2l-1}+x_{n+2l}$, $b_{n+2l}=1$, and $u^{l'}_{n+2l} =0,\ \forall l'\in[k]$.
\item[$P_2$.] $\forall l \in [k]$, note that $x_{n+2l}$ appears only in one constraint. Thus $n+2l^{th}$ column of $A$ is a unit vector
with $n+2l^{th}$ co-ordinate set to one, and hence $(A^T\yy)_{n+2l} = y_{n+2l},\ \forall l\in [k]$.
\item[$P_3$.] From ($P_2$) and the ConstructCost procedure it follows that $c_{n+2l}=1,\ \forall l \in [k]$.
Therefore, the non-zero co-ordinate of $\vv^l$, namely $\nfrac{1}{c_{n+2l}}$, is $1$, for all $l\in[k]$.
\item[$P_4$.] Since $H_{ij}=\nfrac{A_{ij}}{c_j}$, we get that $\forall l \in[k]$,
$(H^T\yy)_{n+2l} = y_{n+2l}$ (using ($P_2$) and ($P_3$)), and  $(H'\xx)_{n+2l} \ge b_{n+2l} \equiv
\nfrac{x_{n+2l-1}}{c_{n+2l-1}}+x_{n+2l}\ge 1$
(using ($P_1$) and ($P_3$)).
\end{itemize}

The above properties are crucial to the over all reduction.

\begin{lemma}\label{lem.lcp}
Vector $(\xx',\yy')$ is a solution of $\lcp_C$ of (\ref{eq.lcp}) if and only if $\ll \in [0,\ 1]^k$, where 
$\l_l=x'_{n+2l},\ \forall l \in [k]$, is a fixed-point of $F^{lp}$.  
\end{lemma}
\begin{proof}
($\Rightarrow$)
Let $(\xx',\yy')$ be a solution of $\lcp_C$. Then by construction of $\lcp_C$, clearly $\xx'$ and $\yy'$ are solutions of
$LP'(\ll)$ and $DLP'(\ll)$ respectively, where $\l_l=\vv^{l^T}\xx'=x'_{n+2l},\ \forall l \in [k]$ (using $P_3$). 
Set $\yy=\yy'$, and $\xx$ be such that $x_j = \nfrac{x'_j}{c_j}$, then using Lemma \ref{lem.equiv} we get that, $\xx$ and
$\yy$ are solutions of $LP(\ll)$ and $DLP(\ll)$. Further, property $P_3$ ensures that $x_{n+2l}=x'_{n+2l}=\l_l, \forall l \in [k]$.
Thus, $\ll$ is a fixed-point of $F^{lp}$.

($\Leftarrow$) Let $\ll$ be a fixed-point of $F^{lp}$ and let $\xx$ and $\yy$ be the solutions of $LP(\ll)$ and $DLP(\ll)$.
Let $\yy'=\yy$ and $x'_j = c_jx_j, \forall j \in [m]$, then using Lemma \ref{lem.equiv} we get that $\xx'$ and $\yy'$ are
solutions of $LP'(\ll)$ and $DLP'(\ll)$ respectively. 
Using the fact that $\ll$ is a fixed-point of $F^{lp}$ and property $P_3$, we get $\vv^{l^T}\cdot\xx'=x'_{n+2l}=x_{n+2l}=\l_l,\
\forall l \in [k]$. In that case, feasibility and complementary slackness of $LP'(\ll)$ and $DLP'(\ll)$, ensures that
$(\xx',\yy')$ is a solution of $\lcp_C$.
\end{proof}

Next, we capture solutions of $\lcp_C$ as Nash equilibria of a bimatrix game.
Consider the following game:
\begin{equation}\label{eq.game}
\begin{array}{ll}
\tilde{A} = \left[\begin{array}{cc} H^T & \zeros \\ \zeros^T & 1\end{array}\right],\ \ \ \ \ \ \ \ 
&
\tilde{B} = \left[\begin{array}{cc} -H'^T & \zeros \\
\bb^T+\ones^T & 1\end{array}\right]
\end{array}
\end{equation}
where $\ones$ and $\zeros$ are $m$-dimensional vectors of $1$s and $0$s respectively.  Number of strategies of both the
players is $m+1$. Let $(\txx,s)$ and $(\tyy,t)$ denote mixed-strategy vectors of the first player and the second player,
then we have, 
\begin{equation}\label{eq.ne1}
(\txx,s)\ge 0;\ \ \ (\tyy,t)\ge0;\ \ \ s+\sum_{i=1}^m \tx_i=1;\ \ \ t+\sum_{j=1}^m \ty_j=1\\
\end{equation}

\begin{remark}
Adler and Verma \cite{adler}, used this idea of adding an extra column/row to handle the r.h.s., in their reduction from `solving' some
special LCPs to symmetric game.  \end{remark}

The property that matrix of $\lcp_C$ is semi-monotone, shown in the next lemma, is important to derive equivalence
between the NE of $(\tA,\tB)$ and the solutions of $\lcp_C$.

\begin{lemma}\label{lem.semimon}
Let $M=\left[\begin{array}{cc} \zeros & H^T \\ -H' & \zeros\end{array}\right]$ be the matrix of $\lcp_C$. For any
$\qq\in \R^{2m}$ with $\qq>0$, the only solution of LCP $\{M\zz \le \qq;\ \zz\ge 0;\ \zz^T(M\zz-\qq)=0\}$ is $\zz=0$.
\end{lemma}
\begin{proof}
It suffices to show that for any $\zz \ge 0, \zz\neq 0$, there is a $d \in [2m]$ such that $z_d>0$ and $(M\zz)_d\le 0$.
Partition $\zz$ as $(\xx,\yy)$.
If $\forall l\in [k],\ z_{n+2l}=x_{n+2l}=0$, then $H'\xx=H\xx$. Therefore, $\zz^TM\zz=\xx^TH^T\yy - \yy^TH\xx =0$. For all
$d \in [m]$, if we have, $z_d>0 \Rightarrow (M\zz)_d >0$ then $\zz^TM\zz>0$, a contradiction.

On the other hand, $\exists l \in [k]$ with $z_{n+2l}>0$ and $(M\zz)_{n+2l}\le 0$ then done. Otherwise, we have $z_{n+2l}>0$
and $(M\zz)_{n+2l}> 0$. This gives $(M\zz)_{n+2l}=y_{n+2l}=z_{m+n+2l}>0$ and \\
$(M\zz)_{m+n+2l}=-(H'\xx)_{n+2l}=-\nfrac{x_{n+2l-1}}{c_{n+2l-1}}-x_{n+2l}= \nfrac{x_{n+2l-1}}{c_{n+2l-1}}-z_{n+2l}<0$ (Using
($P_4$)).
\end{proof}

If $((\txx,s),(\tyy,t))$ is a Nash equilibrium of game $(\tilde{A},\tilde{B})$, the following have to be satisfied (see
Lemma \ref{lem.nash} for the NE characterization), where $\pi_1$ and $\pi_2$ are the scalars capturing payoffs of the first
and the second player respectively.

\begin{equation}\label{eq.ne2}
\begin{array}{lll}
& t \le \pi_1; & s(t-\pi_1)=0 \\
& s \le \pi_2; & t(s-\pi_2)=0 \\
\forall i \in [m],& (H^T\tyy)_i \le \pi_1;& \tx_i((H^T\tyy)_i-\pi_1)=0  \\ 
\forall j \in [m], & (-\txx^TH')_j + b_j s +s \le \pi_2;  &  \ty_j((-\txx^TH')_j+ b_j s +s-\pi_2)=0 \\ 
\end{array}
\end{equation}

\begin{lemma}\label{lem.netolcp}
If $((\tx,s),(\ty,t))$ is a Nash equilibrium of game $(\tilde{A},\tilde{B})$ with $s>0$ and $t>0$, then
$(\frac{\txx}{s},\frac{\tyy}{t})$ is a solution of $\lcp_C$. Further, if $(\xx,\yy)$ is a solution of $\lcp_C$ then
$(\frac{(\xx,1)}{1+\sum_i x_i}, \frac{(\yy,1)}{1+\sum_i y_i})$ is a NE of game $(\tilde{A},\tilde{B})$.
\end{lemma}
\begin{proof}
Since, $s>0$ and $t>0$, we have $\pi_1=t$ and $\pi_2=s$ respectively (using (\ref{eq.ne2})). 
Replacing $\pi_1$ and $\pi_2$ accordingly in the inequalities of (\ref{eq.ne2}), we get
\[
\begin{array}{lll}
\forall i \in [m],& (H^T\tyy)_i \le t;\ \ \ \ & \tx_i((H^T\tyy)_i-t)=0 \\ 
\forall j \in [m], & (H'\txx)_j \ge  b_j s; \ \ \ \ \ \ \ \ &  
\ty_j((H'\txx)_j-b_js)=0 \\ 
\end{array}
\]
Dividing the first expression of first line by $t$ and of second line by $s$, and the second expression in both lines by $s
* t$, we get constraints of $\lcp_C$.Thus $(\frac{\txx}{s},\frac{\tyy}{t})$ is a solution of the LCP.
The second part is easy to verify using the formulation of $\lcp_C$ (\ref{eq.lcp}) and NE conditions (\ref{eq.ne1}) and (\ref{eq.ne2}).
\end{proof}

Lemma \ref{lem.netolcp} shows that NE of game $(\tA,\tB)$ with $s>0,t>0$ exactly capture the solutions of
$\lcp_C$. Next lemma shows that these are the only NE of this game.

\begin{lemma}\label{lem.tnonzero}
If $((\txx,s),(\tyy,t))$ is a Nash equilibrium of game $(\tilde{A},\tilde{B})$ then $s>0$ and $t>0$. 
\end{lemma}
\begin{proof}
We will derive a contradiction for each of the three cases separately.
\medskip

\noindent{\em Case 1:} $s>0$ and $t=0$ \\
Then, we have $\pi_1=t=0$ and therefore, $H^T\tyy \le 0$. Since $H^T$ is upper-triangular with strictly positive values on the
diagonal, the only solution of $H^T\tyy\le 0$ is $\tyy=0$, which contradicts the fact that co-ordinates of vector
$(\tyy,t)$ sums to one (see (\ref{eq.ne1})).
\medskip

\noindent{\em Case 2:} $s=0$ and $t>0$ \\
Then, we have $\pi_2=s=0$ and therefore, $-H'\txx \le 0$. Recall that $H'=H-\sum_{l=1}^k \uu^l\cdot \vv^{l^T}$ and
$\vv^{l^T}\cdot \txx = \tx_{n+2l}$. Further, due to property ($P_4$), $\forall l \in [k],\
(H'\txx)_{n+2l}=\nfrac{\tx_{n+2l-1}}{c_{n+2l-1}}+\tx_{n+2l}$. And, due to ($P_2$) we have $(H^T\tyy)_{n+2l}=\ty_{n+2l}$. 

Now, for an $l \in [k]$ if $\tx_{2+nl}>0$, then $(H^T\tyy)_{n+2l}=\pi_1 \Rightarrow \ty_{n+2l} = \pi_1 > 0$
(using (\ref{eq.ne2}) and $t>0$). However, the $n+2l^{th}$ strategy of the second player is not fetching the maximum
payoff, because $(-H'\txx)_{n+2l}\le -\tx_{n+2l}<0$, a contradiction. 

Thus, we have $\tx_{2+nl}=0,\forall l \in [k]$. Then $H'\txx=H\txx$. Further, the best response condition of the first player
gives $(\txx,s)^T\tA(\tyy,t)=\pi_1 \Rightarrow \txx^TH^T\tyy = \pi_1>0$, and the best
response condition of the second player gives $(\txx,s)^T\tB(\tyy,t)=\pi_2 \Rightarrow
\txx^TH^T\tyy =0$ a contradiction.
\medskip

\noindent{\em Case 3:} $s=0$ and $t=0$ \\
If $\pi_1>0$ and $\pi_2>0$, then due to conditions (\ref{eq.ne1}) and (\ref{eq.ne2}), vector $\tzz=(\txx,\tyy)\neq 0$ is a
solution of LCP $M\zz\le \qq, \zz\ge0, \zz^T(M\zz-\qq)=0$, where $M=\left[\begin{array}{cc} \zeros & H^T \\ -H' &
\zeros\end{array}\right]$ and $\qq=(\pi_1*\ones,\pi_2*\ones)>0$. This contradicts Lemma \ref{lem.semimon}.

If $\pi_1=0$, then the argument is similar to {\em Case 1}. If $\pi_1>0$ and $\pi_2=0$, then it is similar to
{\em Case 2}.
\end{proof}

Now, we have established all the required facts to obtain the main theorems. 
Using Lemmas \ref{lem.tnonzero}, \ref{lem.netolcp}, \ref{lem.lcp}, \ref{lem.flp}, and 
\ref{lem.size}, we show the next theorem.

\begin{theorem}\label{thm.lftogame}
Given a $kD$-Linear-FIXP function $F$ defined by circuit $C$, there exists a bimatrix game $(\tA,\tB)$ with $rank(\tA+\tB)
\le (k+1)$, and $\tA$ upper-triangular, such that the Nash equilibrium strategies
of the first player in game $(\tA,\tB)$ are in one-to-one correspondence with the fixed-points of function $F$, where $size[\tA]+size[\tB]
\le poly(size[C])$.
\end{theorem}
\begin{proof}
From circuit $C$ of $F$ construct $F^{lp}$ of (\ref{eq.flp}), then $\lcp_C$ of (\ref{eq.lcp}) from $F^{lp}$, and finally
game $(\tA,\tB)$ of (\ref{eq.game}) from the LCP.
Using Lemmas \ref{lem.flp} and \ref{lem.lcp} it follows that solution vectors $\xx$ of $\lcp_C$ are in
one-to-one correspondence with the fixed point of $F^{lp}$, which are exactly the fixed points of function $F$
in Linear-FIXP that we started with. 

Further, the Nash equilibrium $(\txx,s),(\tyy,t))$ of game $(\tA,\tB)$ maps to a solution $(\frac{\txx}{s},\frac{\tyy}{t})$ of $\lcp_C$
(due to Lemmas \ref{lem.netolcp} and \ref{lem.tnonzero}). And, two NE with distinct first players strategies $(\txx,s)\neq(\txx',s')$
can not map to the same $\xx$ in a solution of $\lcp_C$. If they do, then we have $\frac{\txx}{s}=\frac{\txx'}{s'} \Rightarrow s'\sum_i
\tx_i =s\sum_i \tx'_i \Rightarrow s'(1-s)=s(1-s') \Rightarrow s'=s \Rightarrow \txx=\txx'$, a contradiction.

Thus we get a game $(\tA,\tB)$ whose Nash equilibrium strategies of the first player are in
one-to-one correspondence with the fixed points of $F$.  Since $H'=H-\sum_{l=1}^k \uu^l\vv^{l^T}$, $rank(\tA+\tB)\le k+1$,
and since $H$ is upper-triangular, $\tA$ is also upper-triangular. The size of matrices $\tA$ and $\tB$ is bounded by
polynomial in size of $A$, $\bb$, $\cc$ and $\uu^l,\ \forall l \in [k]$, and hence the theorem follows using Lemma
\ref{lem.size}.
\end{proof}

Using Theorems \ref{thm.kdlf} and \ref{thm.lftogame}, we get the next theorem.

\begin{theorem}\label{thm.hard}
Nash equilibrium computation in bimatrix games with rank-$k$, $k > 2$ is PPAD-hard.
\end{theorem}

Since, matrix $\tA$ is upper-triangular, we get the following corollary,

\begin{corollary}\label{cor.lu}
Nash equilibrium computation in constant rank bimatrix games with one of the matrix being lower/upper-triangular is PPAD-hard.
\end{corollary}

NE computation in a bimatrix game $(A,B)$ can be reduced to computing a symmetric NE
of a symmetric bimatrix game $(S,S^T)$ where $S=\left[\begin{array}{cc} 0 & A \\ B^T &  0\end{array}\right]$ \cite{agt.ch2}.
Note that if, $rank(A+B)$ is $k$ then $rank(S+S^T)$ is $2k$, and therefore using Theorem \ref{thm.hard} we get,

\begin{corollary}\label{cor.symm}
Computing a symmetric Nash equilibrium of a symmetric game with rank-$k$, $k>5$, is PPAD-hard.
\end{corollary}

Etessami and Yannakakis \cite{EY07} showed that solving a simple stochastic games reduces to computing a unique fixed-point of a
Linear-FIXP problem.  Note that if the Linear-FIXP instance that we start with has a unique fixed-point then the resulting
game in Theorem \ref{thm.lftogame}
will have a unique Nash equilibrium strategy of the first player. In that case, the NE strategies of the second player
should form a convex set because they are essentially solutions of a feasibility lp (follows Lemma \ref{lem.nash}). 
Using this together with the result of \cite{EY07}, we get the following.

\begin{corollary}\label{cor.ssg}
Nash equilibrium computation in bimatrix games with a convex set of Nash equilibria is as hard as solving a simple stochastic game.
\end{corollary}

Chen et. al. \cite{CDT} showed PPAD-hardness for NE computation in bimatrix games ($2$-Nash), which also implies that
symmetric NE computation in symmetric bimatrix game is PPAD-hard (symmetric $2$-Nash) as the former reduces to the latter
(discussed in Section \ref{sec.prel}).  Theorem \ref{thm.hard} gives an alternative proof of these facts.  The Chen et. al.
reduction goes through generalized circuit (similar to Linear-FIXP circuit) with fuzzy gates, graphical games, and game
gadgets to simulate each gate of the generalized circuit separately.  Our reduction bypasses all of these completely, and
provides a simpler reduction using the connections between LPs, LCPs and bimatrix games.
In the next section we give further simplified proof for PPAD-hardness of $2$-Nash and symmetric $2$-Nash, bypassing even
the parameterized LP.

\subsection{Hardness of symmetric and non-symmetric $2$-Nash}\label{sec.simple}
Lemma \ref{lem.max} shows that (\ref{eq.m1}) and (\ref{eq.m2}) are enough to capture execution of the circuit in $\xx$ for
given $\ll$.  As discussed in Section \ref{sec.lp} for any $\ll \in \R^k$ and $\xx \in \CP(\ll)$ (the polyhedron defined in
(\ref{eq.1})), vector $(\xx,\ll)$, satisfies (\ref{eq.m1}).  However, enforcing (\ref{eq.m2}) requires quadratic
complementarity-type constraints. Using these facts, in this section we directly construct an LCP and then a symmetric
bimatrix game (without going through the parameterized LP). This will give a further simplified proofs for PPAD-hardness of
symmetric $2$-Nash, and also for $2$-Nash using the reduction from the former to the later through imitation games
\cite{MT}.

Consider the $A$, $\bb$ and $\uu^l, \forall l \in [k]$ of (\ref{eq.1}).  Recall that $A$, $\bb$ and $\uu^l$ satisfies
properties ($P_1$) and ($P_2$) described in Section \ref{sec.lcp}. Further, since evaluating circuit $C$ is equivalent to
satisfying (\ref{eq.m1}) and (\ref{eq.m2}), where $x_i$ captures the output of $i^{th}$ max gate, $x_{n+2l}, \ \forall l \in
[k]$ captures the outputs of the circuit (Lemma \ref{lem.max}).  Let $\vv^l \in \R^m$ be a unit vector with $1$ on
$(n+2l)^{th}$ co-ordinate and zeros otherwise, i.e., $\vv^{l^T} \xx = x_{n+2l}$.  Let $A'=A-\sum_{l \in [k]} \uu^l
\vv^{l^T}$, and consider the following LCP.

\begin{equation}\label{eq.lcp2}
\begin{array}{lc}
& \xx\ge 0; \ \ \ \  A'\xx\ge \bb \\
\forall i \in [m], & x_i((A'\xx)_i-b_i)=0
\end{array}
\end{equation}

\begin{lemma}\label{lem.lcp2}
Vector $\xx \in \R^m$ is a solution of LCP \ref{eq.lcp2} iff $\ll$, where $\l_l=x_{n+2l}, \forall l \in [k]$, is a
fixed-point of the Linear-FIXP function $F$.  
\end{lemma}
\begin{proof}
Since, $\forall l \in [k],\ \l_l = x_{n+2l}=\vv^{l^T} \xx$ the forward direction follows, because $(\xx,\ll)$
satisfies both (\ref{eq.m1}) and (\ref{eq.m2}) by construction (Lemma \ref{lem.max}).
For the reverse direction let $\ll$ be a fixed-point of $F$ and $\xx$ be a vector such that $x_i$ is the output value of $i^{th}$ max
gate when $\ll$ is the input to the circuit $C$. 
Clearly $(\xx,\ll)$ satisfies (\ref{eq.m1}) and (\ref{eq.m2}) (Lemma \ref{lem.max}). 
The lemma follows using the fact that $\l_l = x_{n+2l} = \vv^{l^T} \xx, \forall l \in [k]$ because $\ll$ is a fixed point.
\end{proof}

Next, we construct a symmetric bimatrix game whose symmetric NE are in one-to-one correspondence with the solutions of LCP
(\ref{eq.lcp2}). Let $S$ be the following $(m+1)\times (m+1)$-dimensional matrix.

\[
S=\left[\begin{array}{cc}
-A' & b+\ones\\
\zeros^T & 1
\end{array}\right]
\]

Consider the symmetric game $(S,S^T)$. Using Lemma \ref{lem.symnash} we get that a mixed strategy vector $\zz=(\xx,t) \in \R^{(m+1)}$ is a
symmetric NE of game $(S,S^T)$ if and only if 

\begin{equation}\label{eq.sne1}
\begin{array}{lcc}
& \xx\ge0;\ \ \  t\ge 0; &  t+\sum_{i \in [m]} x_i=1\\
& -A'\xx+\bb t+t \le \pi;&  t \le \pi; \\
\forall i \in [m],&  x_i((-A'\xx)_i+b_i t+ t-\pi)=0; & t(t-\pi)=0
\end{array}
\end{equation}

where $\pi$ is the payoff $\zz^T S\zz$ of both the agents at NE $(\zz,\zz)$. 

\begin{lemma}\label{lem.sne}
Strategy $\zz=(\xx,t)$, with $t>0$, is a symmetric NE of game $(S,S^T)$ iff $\xx'=\frac{\xx}{t}$ is a solution of LCP
(\ref{eq.lcp2}).  Further, if $\xx$ is a solution of LCP (\ref{eq.lcp2}) then
$\frac{(\xx,1)}{1+\sum_i x_i}$ is a symmetric NE of game $(S,S^T)$.
\end{lemma}
\begin{proof}
If $t>0$ then the third condition of (\ref{eq.sne1}) ensures that $\pi=t$. In that case, the second inequality becomes
$A'\frac{\xx}{t} \ge \bb$, and the third equality becomes $\frac{x_i}{t}(A'\xx-b)_i=0,\ \forall i \in [m]$, which are
exactly the conditions of LCP (\ref{eq.lcp2}). Therefore, $\xx'=\frac{\txx}{t}$ is a
solution of the LCP
(\ref{eq.lcp2}).

Further, if $\xx'$ is a solution of the LCP, then for $t=\frac{1}{1+\sum_i x'_i}$, $x_i=t x'_i$ and $\pi=t$, $((\xx,t),\pi)$
satisfies all the conditions of (\ref{eq.sne1}), and hence the lemma follows.
The second part follows using the conditions of LCP formulation (\ref{eq.lcp2}) and symmetric NE (\ref{eq.sne1}).
\end{proof}

Lemma \ref{lem.sne} shows that symmetric NE of game $(S,S^T)$ with $t>0$ are in one-to-one correspondence with the solutions of LCP
(\ref{eq.lcp2}). One-to-one because clearly no two symmetric NE of game $(S,S^T)$ maps to the same solution
of LCP (\ref{eq.lcp2}). Next, we show that these are the only Symmetric NE of this game.

\begin{lemma}\label{lem.tzero}
If $\zz=(\xx,t)$ is a symmetric NE of game $(S,S^T)$ then $t>0$.
\end{lemma}
\begin{proof}
To the contrary suppose $t=0$, then $\pi\ge t =0$ and $-A'\xx \le \pi$. Recall that $(-A'\xx)_{n+2l}=-x_{n+2l-1}-x_{n+2l},\
\forall l \in [k]$ using ($P_1$). Now, if $x_{n+2l}>0$ then $-(A'\xx)_{n+2l}<0$ which contradicts the third condition of
(\ref{eq.sne1}). Therefore, we have $\forall l\in [k],\ x_{n+2l}=0$ implying that $A'\xx=A\xx$ because $\vv^{l^T} \xx
=x_{n+2l}=0,\ \forall l \in [k]$. Let $i^*$ be the first
strategy played with the non-zero probability, i.e., $i^*=\mbox{argmin}_{x_i>0, i\in [m]} i$. The payoff from $i^*$
should be maximum and hence $(-A'\xx)_{i^*}=\pi$. Since, $A$ is lower-triangular we have
$\pi=(-A'\xx)_{i^*}=(-A\xx)_{i^*} = x_{i^*} <0$, a contradicting $0 = t \le \pi$.
\end{proof}

The next theorem follows using Theorem \ref{thm.kdlf} and Lemmas \ref{lem.size}, \ref{lem.lcp2}, \ref{lem.sne} and
\ref{lem.tzero}.

\begin{theorem}\label{thm.sym2nash}
The problem of computing a symmetric Nash equilibrium of a symmetric bimatrix game is PPAD-hard.
\end{theorem}

As discussed in Section \ref{sec.lcp}, \cite{EY07} showed that solving simple stochastic games \cite{condon} reduces to
finding a unique fixed-point of a Linear-FIXP problem. Using this together with Theorem \ref{thm.sym2nash} we get the next
corollary.

\begin{corollary}\label{cor.ssg_sym}
Computing a unique symmetric NE of a symmetric game is as 
hard as solving a simple stochastic game.
\end{corollary}

McLannen and Tourky \cite{MT} showed that the symmetric Nash equilibria of a symmetric game $(S,S^T)$ are in one-to-one correspondence
with the Nash equilibrium strategies of the second player of game $(S,I)$, where $I$ is an identity matrix. Thus the next theorem
follows using Theorem \ref{thm.sym2nash}.

\begin{theorem}\label{thm.2nash}
The problem of computing a Nash equilibrium of a bimatrix game is PPAD-hard.
\end{theorem}

\section{Linear-FIXP: Hardness of Approximation}\label{sec.approx}
Chen et. al. \cite{CDT} showed that higher dimensional discrete fixed-point problem (defined below) is PPAD-hard even when
the grid has a constant length in each dimension. Using this result, in this section we show inapproximability results for
Linear-FIXP, by reducing a discrete fixed point problems to finding an approximate solution of a Linear-FIXP
problem; the reduction is similar to that of Section \ref{sec.2dlf}. An approximate fixed point can be
defined as follows:

\begin{definition}
Vector $\xx \in [0,\ 1]^k$ is an $\epsilon$-approximate fixed point of function $F:[0,\ 1]^k\rightarrow [0,\ 1]^k$ if $\|\xx
- F(\xx)\|_\infty \le \epsilon$.
\end{definition}

Similar to $2D$-Brouwer, let $kD$-Brouwer represent the class of $k$-dimensional discrete fixed-point problems.
An instance of $kD$-Brouwer consists of a grid $G^k_{n}=\{0,\dots,2^n-1\}^k$, and a valid coloring function
$g:G^k_{n}\rightarrow\{0,1,\dots,k\}$, which satisfies the following: Let $\partial(G^k_{n})$ denote the set of points $\pp \in G^k_{n}$
with $p_i \in \{0,2^n-1\}$ for some $i$, i.e., boundary points, then,
\[
\mbox{For $\pp \in \partial(G^k_{n})$, if } p_i>0, \forall i \in [k]  \mbox{ then } g(\pp)=0, \mbox{ otherwise }
g(\pp)=\max\{i\ |\ p_i=0,
i \in [k]\}
\]

Let $K_\pp=\{\qq\ |\ q_i \in \{p_i,p_i+1\}\}$ be the set of vertices of a unit hyper-cube with $\pp$ at the lowest-corner. As discussed
in \cite{CDT}, given any valid coloring $g$ of $G^k_{n}$, $\exists \pp \in G^k_{n}$ such that the vertices of hyper-cube $K_\pp$ have
all $k+1$ colors; $K_\pp$ is called a panchromatic cube.  However, since there are $2^k$ vertices in a hyper-cube, given $\pp$ there is
no efficient way to check if $K_\pp$ is panchromatic. Therefore, Chen et. al. introduces the following notion of discrete fixed points.

\begin{definition}[Panchromatic Simplex \cite{CDT}]
A subset $P \subset G^k_{n}$ is {\em accommodated} if $P\subset K_\pp$ for some point $\pp \in G^k_{n}$. It is a {\em panchromatic
simplex} of a color assignment $g$ if it is accommodated and contains exactly $k+1$ points with $k+1$ distinct colors.
\end{definition}

From the above discussion it follows that for any valid coloring $g$ on $G^k_{n}$, there exists a panchromatic simplex in
$G^k_{n}$ \cite{CDT}.  Similar to $2D$-Brouwer the coloring function $g$ is specified by a $kD$-Brouwer mapping circuit $C^b$.
\medskip

\noindent{\bf $kD$-Brouwer Mapping Circuit:}
The circuit has $kn$ input bits, $n$ bits for each of the $k$ integers representing a grid point, and $2k$ output bits
$\D_i^+,\D_i^-,\ \forall i \in [k]$. It is a {\em valid Brouwer-mapping circuit} if the following is true:

\begin{itemize}
\item For every $\pp \in G_{n}$, the $2k$ output bits of $C^b$ satisfies one of the following $k+1$ cases: 
\begin{itemize}
\item Case $0$: $\forall i \in [k] $, $\D^-_i=1$ and $\D^+_i=0$.
\item Case $i$, $i \in [k]$: $\D^+_i=1$ and all the other $2k-1$ bits are zero.
\end{itemize}
\item For every $\pp\in \partial{G^k_n}$, if $\exists i \in [k]$ with $p_i=0$ then letting $i_{max}=\max\{i\ |\ p_i=0\}$, the output bits
satisfy {\em Case $i_{max}$}, otherwise they satisfy {\em Case 0}.
\end{itemize}

Such a circuit $C^b$ defines a valid color assignment $g_{C^b}: G^k_{n} \rightarrow \{0,1,\dots,k\}$ by setting $g_{C^b}(\pp)=i$, if
the output bits of $C^b$ evaluated at $\pp$ satisfy Case $i$. 
Let $\ee^i$ be a $k$-dimensional unit vector with $1$ on $i^{th}$
coordinate, and $\ee^0$ be a vector with all $k$ coordinates set to $-1$. 
Then, a $k$-dimensional vector $I$ set to $I_i=\D^+_i - \D^-_i,\ \forall i \in [k]$ is $\ee^i$ for {\em Case $i$}. This defines a 
discrete function $H:G^k_n\rightarrow G^k_n$ where $H(\pp)=\pp+\ee^{g_{C^b}(\pp)}$. 

Given a $kD$-Brouwer mapping circuit $C^b$ on grid $G^k_n$, next we 
construct a $kD$-Linear-FIXP circuit $C$ defining a function $F:[0,\ 2^n-1]^k \rightarrow [0,\ 2^n-1]^k$, which is an extension of
function $H$.  We show that all the $\frac{1}{poly(\CL)}$-approximate fixed-points of $F$ 
are in panchromatic cubes of $G^k_{n}$, where $\CL$ is the size of circuit $C^b$. 
Further, we give a polynomial time procedure to
compute a panchromatic simplex from an approximate fixed-point. When we reduce the range from $[0,\ 2^n-1]^k$ to $[0,\ 1]^k$, to bring
the function in to a standard form of Linear-FIXP, the approximation factor becomes $\frac{1}{2^npoly(\CL)}$.

Recall that circuit $C$ has $k$ real inputs and outputs, $\{\max,+,*\zeta\}$ operations, and rational constants.
The construction is almost same as that in Section \ref{sec.2dlf}. Let $L > k^4 $ be a large integer with value being a power of $2$,
and at most polynomial in $size[C^b]$, {\em i.e.}, $L=2^l \le poly(size[C^b])$. As in Definition \ref{def.well-poor} {\em
well-positioned} and {\em poorly-positioned} points of $\Real_+^k$ may be defined. Further, for $\pp \in [0,\ 2^n)^k$, let
$\pi(\pp)=\lfloor \pp \rfloor$ and $\zeta(\pp)=\ee^{g_{C^b}(\pi(\pp))}$.

For a well-positioned point $\pp \in [0,\
2^n)$ the bit representation of each coordinate of $\pi(\pp)$ can be computed in $C$ using
ExtractBits procedure of Table \ref{tab1} (due to Lemma \ref{lem.be}). This bit representation, when fed to a simulation of $C^b$ where
$\land$, $\lor$ and $\land$ are replaced by $\min$, $\max$ and $(1-x)$ respectively, outputs $2k$ values which is exactly
$C^b(\pi(\pp))$. However, it is still not clear how to efficiently check if hyper-cube $K_\qq$ is panchromatic because it has $2^k$
vertices. Further, if $\pp$ is poorly-positioned to start with, then it is not clear how to compute even the bit representation of
$\pi(\pp)$ using operations of Linear-FIXP.
To circumvent these issues we use a geometric lemma proved by Chen et. al. \cite{CDT}, described next.
For a finite set $S \subset \Real_+^k$, let $I_w(S)$ contain the indices of the well-positioned points of $S$ and $I_p(S)$ contain
indices of poorly-positioned points.

\begin{lemma}\label{lem.sampling2}
\cite{CDT} Given $\pp \in [0,\ 2^n-1]^k$, consider the set $S=\{\pp^1,\dots,\pp^{k^4}\}$ such that
\[
\pp^j = \pp + \frac{(j-1)}{L}\sum_{i\in[k]} \ee^i,\ \ \ \ j \in [k^4] 
\]

For each $j \in I_p(S)$, let $\rr^j \in \R^k$ be a vector with $\|\rr^j\|_\infty \le 1$.
And for each $j \in I_w(S)$, let $\rr^j=\zeta(\pp^k)$.
If $\| \sum_{j=1}^{n^4} \rr^j\|_\infty< 1$ then $Q_w=\{\pi(\pp^j)\ |\ j \in I_w(S)\}$ is panchromatic simplex.
\end{lemma}
\begin{proof}
Let $Q=\{\qq^j=\pi(\pp^j) \ |\ \pp^j \in S\}$. Since $\frac{k^4}{L} <<1$ the set crosses boundaries of the unit cells at most
$k$ times. In other words, for each $i \in [k]$, there is at most one $j_i$ such
that $q_i^{j_i} = q_i^{j_i-1}+1$. Therefore, set $Q$ can have at most $k+1$ elements, and they are part of the same unit hyper-cube,
which has to be $K_{\pi(\pp)}$. Clearly, $Q_w \subset Q$.

Further, since $\frac{1}{L^2}<<\frac{1}{L}<<1$, there can be at most $k$ poorly-positioned points in $S$. So, we have
$|I_w(S)|\ge k^4-k$. Let $\rr^G = \sum_{j\in I_w(S)} \rr^j$, then we have
$\|\rr^G + \sum_{j \in I_p(S)} \rr^j\|_\infty < 1 \Rightarrow \|\rr^G\|_\infty < 1+\|\sum_{j \in I_p(S)} \rr^j\|_\infty < k+1$, 
because $|I_p(S)|\le k$, and $\|r^j\|_\infty \le 1$ for each $j \in I_p(S)$.

Let $\forall i \in [k],\ W_i$ be the number of indices of $I_w(S)$ with $\rr^k=\ee^i$. Using the above fact, we will show that $W_i\neq
0, \forall i$, to prove the lemma.  

If $W_0=0$ then $W_i > k^2$ for some $i \in [k]$. In that case, $r^G_i\ge k^2$, a contradiction. 
If $W_t=0$ for a $t \in [k]$, then $W_0 < k+1$ or else $r^G_t \ge k+1$. Let $i^* = \argmax_{0 \le i \le k} W_i$, then clearly, $W_{i^*}
\ge k^3-1$ and $i^*\neq 0$. Then, $r^G_{i^*} \ge k^3-1-k$, again a contradiction. 
\end{proof}

Using Lemma \ref{lem.sampling2} we can construct circuit $C$ as done in steps $(S_1)$ to $(S_6)$ in Section \ref{sec.2dlf}, where instead
of $16$, $k^4$ points have to be sampled, and finally in step $(S_6)$ the incremental vector $\sum_j \rr^j$ has to be divided by $k^4$
in order to take an average. This circuit will define a piecewise-linear function $F:[0,\ 2^n-1]^k \rightarrow [0,\
2^n-1]^k$. Next, we show that it suffices to compute a $\frac{1}{L}$-approximate fixed point of $F$ in order to find a panchromatic
simplex.

\begin{lemma}\label{lem.approxfp}
Every $\frac{1}{L}$-approximate fixed point of $F$ is in a panchromatic hyper-cube of $G_{n,k}$.
\end{lemma}
\begin{proof}
Let $\pp$ be a $\frac{1}{L}$-approximate fixed point of $F$, and $\pp'=F(\pp)$.
Then, the set $S$ of sampled points is $S=\{\pp^j =\pp +\frac{(j-1)}{L}\sum_{i\in[k]} \ee^i\ |\ j \in [k^4]\}$, and $\rr^j$ is the
outcome vector in step $(S_4)$ for $\pp^j$. By construction, we have $\rr^j=\zeta(\pp^j), \forall j \in I_w(S)$, and $\|\rr^j\|_\infty
\le 1,\ \forall j \in [k^4]$. Further, $\rr$ is the average of $\rr^j$s, and hence $\|\rr\|_\infty\le 1$.

Suppose, $\pp$ is not inside a panchromatic hyper-cube, then $\|\rr\|_\infty \ge\frac{1}{k^4} > \frac{1}{L}$ by Lemma \ref{lem.sampling2}.
If $\pp$ is at least $\frac{1}{L}$ distance away from the boundary of
$[0,\ 2^n-1]^k$, i.e., $\frac{1}{L} \le p_i \le 2^n-1-\frac{1}{L},\ \forall i \in [k]$, then clearly $\|\pp-F(\pp)\|\ge \frac{1}{L}$, a
contradiction. 

For the points near boundary it may happen that $\|\rr\|_\infty \ge\frac{1}{k^4}$, but still due to rounding in step $(S_6$), they
generate dummy fixed-points. Using the fact that $C^b$ generates a valid coloring, we show that this can never happen.
Let $\pp$ be such that for some $i$ either $p_i<\frac{1}{L}$ or $p_i>2^n-1-\frac{1}{L}$.

\begin{itemize}
\item $\exists i \in [k],\ p_i < \frac{1}{L}$: Let $i_{max}=\max\{i\ |\ p_i < \frac{1}{L}\}$, then $\forall j \in [k^4]
p^j_{i_{max}} <1$. Therefore, $\exists i' \ge i_{max}$ such that $p_i < 1$ and $r_{i'}>0$, implying that $p'_i > p_i$.
\item $\forall i \in [k], p_i>1$: Since $\exists i'$ with $p_{i'}>2^n-1-\frac{1}{L}$, except for $\pp^1$ all other $\pp^j$ are outside
of $[0,\ 2^n-1]^k$, and $\pi(\pp^j)>0$. Therefore, $\forall j \in I_w(S), j \neq 1$, we have $\rr^j=\ee^0<0$. Hence $\exists i$, such
that $p'_i<p_i$.
\item $\exists i, i' \in[k], p_i<1$ and $p_{i'}>2^n-1-\frac{1}{L}$: Let $i_{max}=\max\{i\ | p_i <1\}$, then $\exists i'' \le i_{max}$
such that either $p_{i''}<1$ and $r_{i''}>0$ implying that $p'_{i''}>p_i$, or $r_{i'}<0$ implying that $p'_{i'}<p_{i'}$.
\end{itemize}
\end{proof}

If $\pp$ is a $\frac{1}{L}$-approximate fixed point of $F$, then it is in panchromatic hyper-cube of $G_{n,k}$ (Lemma
\ref{lem.approxfp}), and the panchromatic simplex containing $\pp$ is $\{\pi(\pp^j)\ |\ \pp^j = \pp + \frac{(j-1)}{L}\sum_{i\in[k]}
\ee^i\}$ (Lemma \ref{lem.sampling2}). Therefore given a $\frac{1}{L}$-approximate fixed point of $F$ a panchromatic simplex
of $G^k_n$ can be computed in polynomial time.

We can shrink the range of function $F$ from $[0,\ 2^n-1]$ to $[0,\ 1]^k$ by multiplying and dividing the inputs and outputs
respectively by $2^n-1$. 
For the modified function, $\frac{1}{2^n L}$-approximate fixed-points are guaranteed to be in panchromatic hyper-cubes.
Note that, Lemmas \ref{lem.sampling2} and \ref{lem.approxfp} holds for any $L$ strictly greater than $k^4$, hence
$\frac{1}{2^nL}=\frac{1}{2^npoly(k)}$.
Further, by construction $size[C]$ = (\#inputs + \# gates + total size of the constant used in $C$), 
is polynomial in $size[C^b]$

Therefore, a $kD$-Brouwer problem of computing a panchromatic simplex reduces to finding a
$\frac{1}{\gamma poly(k)}$-approximate fixed-point of a $kD$-Linear-FIXP function, where $\gamma$ is the largest absolute 
constant used in the circuit.
Chen et. al. \cite{CDT} proved that $kD$-Brouwer with $n=3$ and $k$ not a constant is PPAD-hard (Brouwer$^{f_1}$ in \cite{CDT}). 
Since, the largest absolute constant used in the $kD$-Linear-FIXP circuit constructed from such an instance is of $O(1)$, the
next theorem follows,

\begin{theorem}\label{thm.approx}
Let $F$ be a piecewise-linear function defined by a Linear-FIXP circuit $C$, and let $\CL=size[C]$. Then
Computing a $\frac{1}{poly(\CL)}$-approximate fixed-point of $F$ is PPAD-hard. 
\end{theorem}

\begin{remark}
We note that Theorem \ref{thm.approx} may also follow from the $2$-Nash to Linear-FIXP reduction shown by Etessami and
Yannakakis \cite{EY07}.
\end{remark}

\section{Discussion}\label{sec.disc}
In this paper we show that Nash equilibrium computation in bimatrix games with rank$\ge 3$ is PPAD-hard by reducing
$2D$-Brouwer to rank-$3$ games. Given an instance of $2D$-Brouwer first we reduce it to $2D$-Linear-FIXP, a $2$-dimensional
fixed-point problem defined by a Linear-FIXP circuit with two inputs. Next we replace the circuit by a parameterized linear
program with two parameters, and finally using the connections between LPs and LCPs, and LCPs and bimatrix games, we
construct a rank-$3$ game. This, last step of the reduction uses the fact that the parameterized linear program was
constructed from a Linear-FIXP circuit. It will be interesting to reduce a fixed-point problem, defined by a
parameterized LP, to a bimatrix game in general. This will extend the classical construction of zero-sum games from linear
programs by Dantzig \cite{dantzig}. 
If fixed-point problem with $k$-parameter LP can be reduced to a rank-$k$ game, then it will
imply that rank-$2$ games are also PPAD-hard, settling the only unresolved case.

As corollaries of our reduction, we get that $2D$-Linear-FIXP = PPAD = Linear-FIXP, and in turn a sharp dichotomy on
complexity of Linear-FIXP problems; $1D$-Linear-FIXP is in P, while for $k\ge 2$, $kD$-Linear-FIXP is PPAD-complete. 
We also give an explicit construction of a
rank-$(k+1)$ game from a $kD$-Linear-FIXP problem. This construction (almost) preserves the number of solutions (in
terms of NE strategies of the first player), and is different from all the previous approaches. This should be of
useful to understand the connections between problems that reduces to Linear-FIXP, and bimatrix games.
One such example is Corollary \ref{cor.ssg} which shows that even if Nash equilibrium set of a bimatrix game is guaranteed to
be convex, finding one is as hard as solving a simple stochastic game.

In Section \ref{sec.approx} we show hardness of approximation for Linear-FIXP problems. It is not immediately clear how to
extend this to bimatrix games through our reduction.  If done, it will provide an alternate (simpler) proof of
inapproximability in $2$-Nash \cite{CDT}.

For the case of symmetric games, the PPAD-hardness of rank-$3$ games imply that computing a symmetric Nash equilibrium in a
symmetric rank-$6$ games is PPAD-hard. The polynomial time algorithm for computing symmetric NE of a rank-$1$ symmetric
games by Mehta et. al. \cite{MVY} leaves the status of symmetric games with rank-$2$ to rank-$5$ unresolved.
\medskip
\medskip
\medskip

\noindent{\bf Acknowledgments.}
I wish to thank Vijay V. Vazirani, Milind Sohoni, Bernhard von Stengel, Mihalis Yannakakis, and Jugal Garg 
for valuable discussions. 

\bibliographystyle{abbrv}
\bibliography{kelly,sigproc}

\end{document}